\documentclass[12pt,pra, reprint,twocolumns,showpacs,superscriptaddress, nobalancelastpage,notitlepage,nofootinbib,longbibliography]{revtex4-2}
\usepackage{graphicx,color,mathbbol,amsthm,amsmath,amsfonts,amssymb,bm,dsfont,footmisc,mathrsfs,multirow,latexsym,times}

\usepackage{amsmath,amssymb,amsfonts}

\usepackage{graphicx}
\usepackage{bbm}
\usepackage{mathrsfs}
\usepackage[hidelinks]{hyperref}
\usepackage{algpseudocode}
\usepackage{algorithm}
\usepackage[dvipsnames]{xcolor}
\usepackage{tikz}
\usepackage{tikz-network}
\usepackage{enumerate}
\usepackage{caption}
\usepackage{pgfplots}
\usepackage{pgfplots}

\definecolor{myblue}{HTML}{0072BD}
\definecolor{myorange}{HTML}{D95319}
\definecolor{myyellow}{HTML}{EDB120}
\definecolor{myred}{HTML}{E22B12}

\newtheorem{lemma}{Lemma}

\newtheorem{definition}{Definition}
\newtheorem{problem}{Problem}
\newtheorem{proposition}{Proposition}
\newtheorem{remark}{Remark}
\newtheorem{assumption}{Assumption}
\newtheorem{corollary}{Corollary}
\algnewcommand\algorithmicforeach{\textbf{for each}}
\algdef{S}[FOR]{ForEach}[1]{\algorithmicforeach\ #1\ \algorithmicdo}
\newcommand{\Nc}{\mathcal{N}}
\newcommand{\Ac}{\mathcal{A}}
\newcommand{\Bc}{\mathcal{B}}

\newcommand{\Dc}{\mathcal{D}}

\newcommand{\Oc}{\mathcal{O}}
\newcommand{\Hc}{\mathcal{H}}
\newcommand{\Jc}{\mathcal{J}}
\newcommand{\Rc}{\mathcal{R}}

\newcommand{\Sc}{\mathcal{S}}

\newcommand{\Zc}{\mathcal{Z}}

\newcommand{\Lc}{\mathcal{L}}
\newcommand{\Tc}{\mathcal{T}}

\newcommand{\Uc}{\mathcal{U}}

\newcommand{\Xc}{\mathcal{X}}

\DeclareMathAlphabet{\mathpzc}{OT1}{pzc}{m}{it}

\newcommand{\Rb}{\mathbb{R}}
\newcommand{\Cb}{\mathbb{C}}

\newcommand{\Eb}{\mathbb{E}}

\newcommand{\Nb}{\mathbb{N}}

\newcommand{\Hf}{\mathfrak{H}}

\newcommand{\ones}{\mathbbm{1}}

\newcommand{\tr}{{\rm tr}}
\newcommand{\diag}{{\rm diag}}
\newcommand{\Span}{{\rm span}}

\DeclareMathOperator{\rank}{rank}

\newcommand{\ket}[1]{\left| #1 \right>}
\newcommand{\bra}[1]{\left< #1 \right|}
\newcommand{\norm}[1]{\left|\left| #1 \right|\right|}

\definecolor{darkviolet}{rgb}{0.58, 0.0, 0.83}
\definecolor{lavender}{rgb}{0.45, 0.31, 0.59}

\newcommand{\vect}{{\rm vec}}

\newcommand{\vv}[1]{{#1}} 

\begin{document}

\title{\LARGE \bf {Reconstructing Quantum States and Expectations\\ via Dynamical Tomography}} 

\author{Marco Peruzzo} 
\affiliation{Department of Information Engineering, University of Padova, Italy.} 
\author{Tommaso Grigoletto} 
\affiliation{Department of Information Engineering, University of Padova, Italy.} 
\author{Francesco Ticozzi}
\affiliation{Department of Information Engineering, University of Padova, Italy.} 
\affiliation{Department of Physics and Astronomy, Dartmouth College, Hanover, New Hampshire 03755, USA}

\begin{abstract}
When the dynamics of a quantum system of interest is known, an informationally-complete set of observables is not needed for state reconstruction via tomographic techniques: letting the system evolve before performing the measurement allows one to effectively extend the available ways to probe the system.
This idea leads to {\em dynamical quantum tomography}, whose feasibility we characterize for general quantum dynamics using Krylov-based methods. Specializing to Markovian ones, we also provide deterministic tests, and randomized ones to effectively assess parametric dynamics. The limits of the methods are explored comparing unitary and open dynamics when a single observable is available, and the set of observables whose expectation can be reconstructed from the available ones characterized. The framework is illustrated with applications to a spin chain (with or without dissipation) and an electron-nuclear system.
\end{abstract}

\date{\today}

\maketitle
\section{Introduction}

Reconstructing unknown properties and states of quantum systems via repeated measurements,  tomographic protocols and statistical methods are arguably among the most fundamental and crucial tasks in every quantum information experimental setup \cite{holevo,paris-estimationbook,quantumverificationreview}. In order to be able to fully reconstruct a state, typically an {\em informationally complete} set of observables is assumed to be available experimentally. However, for large systems, i.e. multipartite system comprising a large number of identical components, the number of the observables increases prohibitively. 

One way to reduce the need for experimental resources is provided by pre-existing or engineered dynamics on the system: in Heisenberg picture, if we let the available observables evolve, they can potentially provide data on new parts of the system's state space. In practice, the quantum system is subjected to a known time evolution for a certain period of time, after which a measurement of a chosen observable is performed. This procedure may be repeated for multiple periods of time and with different observables.
The use of dynamics to enhance the possibility of tomographic protocol has been named {\em Dynamical Quantum (State) Tomography} (DQST) in \cite{dynamicaltomography}, where the first systematic study of feasibility has been presented. 

\textit{State of the art} - While the literature on standard quantum tomography is enormous, the contributions that focus on dynamical tomography are only a few and we quickly review in the following. In the first work \cite{dynamicaltomography}, it is shown that, for almost all discrete-time unitary evolutions, almost all Positive-Operator Valued Measurements (POVM) with $d+1$ outcomes, defined in a  $d$-dimensional Hilbert space, are sufficient to reconstruct the quantum state via DQST. 
When we allow for general discrete-time evolutions for the quantum system, POVMs with only two outcomes may suffice for the state reconstruction. 
We highlight that in their work, the authors also provide bounds on the number of required measurements for DQST when prior knowledge about the state is available.

In \cite{merkel2010random} the authors show that  when $d>2$ reconstruction of every state is not possible from measurements of a single observable evolving in discrete-time under repeated applications of unitary maps. A state estimation protocol based on continuous weak measurements of the time series of operators generated by the evolution of the latter observable is discussed\footnote{Note that this protocol do not exactly coincide with the dynamical tomography protocol discussed above.}. Furthermore, numerical results show how it is still possible to reconstruct both pure and mixed states with high fidelity. The latter possibility is related to the existence of positivity constraints imposed by the set of physically admissible quantum states.

In \cite{HR-MW:2025}, both tomography of states and observables are considered. By exploiting different tools with respect to the previously mentioned work, the authors find the same no-go result of \cite{merkel2010random}. Furthermore, both the cases of discrete-time and continuous-time general Markovian dynamics are considered, and it is proven that for almost all possible dynamics DQST from a single given observable is feasible.

Recently, a connection between DQST and observability analysis in control theory was established \cite{MP-TG-FT:2024, XS-WY:2024}.
In \cite{XS-WY:2024} the authors consider closed quantum systems undergoing known discrete-time unitary dynamics. The no-go result of \cite{merkel2010random} was proven using well-known tools from systems and control theory \cite{kalman1969topics}. Furthermore, a DQST protocol to estimate the state of the system from measurement-data is proposed.

The tomography problem, i.e. state reconstruction from observable averages, is connected - and can be seen in fact as a generalization - of the so-called quantum marginal problem, where the available data are the state marginals with respect to a given multipartite and locality structure \cite{linden_almost_2002,Linden2002_2,weisQuantumMarginalsFaces2023,karuvade2019uniquely,Xin2017}.  
Dynamical tomography can be used in that setting as well, in general extending the set of states that are uniquely reconstructable given their marginals. A preliminary set of results in this sense has been presented in \cite{MP-TG-FT:2024}, where the notion of states that are Uniquely Dynamically Determined among All (UDDA) states has been introduced: {these are states that can be uniquely reconstructed from a set of available measurements performed on identical preparations of such state, also allowing for {\em arbitrary-time evolution} with a known (closed or open) dynamical model.}  Observability analysis in control theory is exploited to assess when all {states} are UDDA, and to provide a reconstruction protocol via linear inversion.

\textit{Contribution of this work} - In this work, we  outline a general framework for DQST, from feasibility analysis to {practical} methods to select the observables and the evolution times. In doing so, we highlight the role and use the tools of {\em observability analysis} in system theory \cite{kalman1969topics}, a Krylov-subspace approach  that provides all the tools needed to assess whether reconstruction problems for quantum systems are solvable. 

More in detail, in this work:

\begin{itemize} 
\item In Section \ref{sec:DQST}, we formalize the connection between the \textit{feasibility of DQST} for a system, the requirement that \textit{every state of a system be UDDA}, and the concept of an \textit{observable system} in the sense of linear system theory \cite{kalman1969topics}. We prove that these three conditions are indeed  equivalent. The needed tools are recalled in Appendix \ref{sec:PBH}. In the literature, similar results were provided only for the cases of Markovian dynamics, while the conditions we give for feasibility of DQST hold for general Completely-Positive, Trace-Preserving (CPTP) evolutions. Later, we specialize the results to time-homogeneous Markovian dynamics, whose additional structure can be exploited to further simplify the test for feasability of  DQST. 

\item For {\em parametric} Markovian dynamics, we show in Section \ref{sec:parametric} that feasibility of DQST is generic if it holds for at least one parameter set. This allows one to sample dynamics with random parameters for testing a system, or looking to find the largest set of reconstructable observable.

\item We revisit the feasibility analysis when the dynamics is unitary, {\em providing a lower bound for the minimal number of observables needed for DQST} given Hamiltonian dynamics in multipartite systems, as well as showing that a generic unitary dynamics is sufficient to ensure DQST when the lower bounds are met. Furthermore, we show that {\em a class of dissipative Lindblad dynamics is able to ensure DQST from a single observable for networks of qubits.} While it has been shown using different, more abstract methods that generic open dynamics suffice \cite{HR-MW:2025}, our analysis exploits properties of the Pauli basis and shows that generic, purely dissipative generators are sufficient. The full proof, being a rather technical detour, is presented in Appendix \ref{sec:nqubitDQST}.

\item {We propose a {\em protocol to select the observables and times involved in the data-acquisition phase necessary for the state reconstruction}. To this aim, in Section \ref{sec:alg} we outline a heuristic procedure that addresses this problem and selects a minimum-cardinality set of observables, together with times (among the available ones) at which measurement should be performed on the system. The obtained sequence aims to maximize the signal-to-noise ratio of the measurements and the quality of the estimate by selecting observables that become the ``most orthogonal'' to the already selected ones. Albeit sub-optimal, the proposed procedure significantly reduces the computational burden by introducing iterative selection, rather than exploring all possible observable sequences, and represents the first attempt to address a problem that was left as open in the literature.}

\item We introduce and analyze an additional, related problem: {\em predicting expectation of target observables that cannot be directly measured} on the system. Section \ref{sec:obs_exp_reconstr} is devoted to introducing the problem and showing that this problem is solvable, without necessarily reconstructing the full state of the system, provided the target observables belong to specific operator subspaces. We highlight again how the tools of observability analysis are key in deriving the solution.

\item Lastly, we consider {\em multipartite quantum systems}, specializing some results to this case and providing some useful bounds on the minimal number of observables needed to obtain DQST for many-body Hamiltonian systems in Section \ref{sec:multi}, and insight on the structure of observable dynamics. Furthermore, we show how the proposed observability-based framework can be effectively used in some relevant examples to estimate the state of the system and observable expectations: we consider a spin chain with different type of Hamiltonian and dissipative dynamics and a bi-partite system that emerges in the modeling of Nitrogen-Vacancy (NV) centers in diamonds. In particular, we are able to prove that the state of networks of 4 or more qubits are never reconstructable from single-site observables with Hamiltonian dynamics, while in principle bipartite quantum systems are, when one considers local  measurements on each subsystem and an generic Hamiltonian. 
\end{itemize}
These findings highlight the beneficial role of dissipation for DQST, as it can allow for multi-qubit DQST from single-site measurements, and even a single measurement, whereas Hamiltonian cannot.

\section{A Dynamical Approach\\ to Quantum State Tomography}\label{sec:DQST}

\subsection{Notations and basic assumptions}\label{sec:notions_and_assumptions}
In this paper, we only consider finite dimensional quantum system hence we focus on finite dimensional Hilbert spaces $\Hc\simeq\Cb^d$. 
In the following the symbols $\Bc(\Hc), \Hf(\Hc),$ and $\Hf_0(\Hc)$ denote the space of (bounded) linear operators acting on $\Hc$, the subspace of hermitian operators, the subspace of hermitian traceless operators, respectively.

The state of the system is associated to a density operator $\rho$  in $\Dc(\Hc)=\{\rho\in\Bc(\Hc) | \rho=\rho^{\dag}\geq 0, \tr(\rho)=1\}$. In most of this work, we shall assume that it is a priori \textit{unknown}, 
and that we have access to an experimental setup that provides us access to only partial information about the system's state.
More specifically, we consider to be able to perform measurements of a set of Hermitian observables $\Xc=\{X_i\}_{i=0}^l$ on identical copies of system prepared in the unknown state $\rho$.  

 Each observable can be decomposed as $X_i=\sum_k \alpha_{i,k}\Pi_k$, where $\{\Pi_k\}$ is the associated spectral family of orthogonal projectors. The $\{\alpha_{i,k}\}$ are the possible outcomes of the measurement of the $i$-th observable on the system and the quantity $p_k=\tr(\Pi_k \rho)$ represents the probability of obtaining the $k$-th outcome of the measurement. The expectation value of the observable $X_i \in \Xc$ will be labeled as  $$y_i=\mathbb{E}_{\rho}(X_i)=\tr(X_i \rho).$$
 The expectation values can be reconstructed via empirical averages of repeated experiments on identical preparations (see the Appendix \ref{sec:empav} for details.
 
In this work, we  always consider the identity~operator~$X_0\!=\!I$~belong~to~$\Xc$: note that including the identity operator in the set of available observables $\Xc$ is a technical and not restrictive assumption, as measuring the identity is effectively equivalent to saying that some measurement result have been obtained from the system.
Furthermore, we assume $\Xc$ to be a set of linearly independent operators, so that $|\Xc|=\dim(\Span\{\Xc\})$.

\subsection{Dynamical quantum state tomography}
One of the key tasks in quantum information technologies is the state reconstruction of a system of interest, also known as tomography:\smallskip
\begin{problem}[Tomography]\label{prob:tom}
    Assuming to be able to prepare a quantum system in a given and unknown state $\rho$ and to perform single-shot measurements of a set of observables $\Xc\subset\mathfrak{H}(\mathcal{H})$, uniquely reconstruct $\rho$ from the measurement data. 
\end{problem}\smallskip
If no prior information or resources, in addition to the possibility of measuring the observables in ${\cal X}$ are available, it is possible to uniquely determine the state of the system from multiple measurements performed on it 
if and only if $\Xc$ is informationally complete, that is $\forall \rho_1, \rho_2 \in \Dc(\Hc)$ s.t. $\rho_1\neq \rho_2,$ then $\exists X_i \in \Xc$ s.t. $ \ \Eb_{\rho_1}(X_i) \neq \Eb_{\rho_2}(X_i)$. In particular the set $\Xc$ is informationally complete if it generates $\Bc(\Hc)$. The problem is then solved by estimating the expectation values $y_i=\mathbb{E}_\rho(X_i)$ of a linearly independent set of the measured operators $X_i\in {\cal X}$ as described in the previous section, and by linear inversion or a variational method (maximum likelihood, minimum relative entropy, etc) \cite{paris-estimationbook,zorzi2014,zorzi2014minimal}. For this reason, from now on, we focus on the expectations $y_i$ rather than the single-shot measurement data.

However, an informationally complete set of observables is not necessary if we can exploit its dynamics.
In this work, we assume to have {\em perfect knowledge} on the quantum dynamics of the system, which we after assume to be the same each the later is prepared in the state $\rho$. If the dynamics is non-trivial, measuring multiple times an observable $X_i\in{\cal X}$  after some given time $t$ and computing the average corresponds effectively to obtaining the average of the time-evolved observable in Heisenberg picture. 

A general, physically admissible evolution of operators 
is  associated to a set of completely positive (CP) and unital maps  $\{\Phi_t\}_{t\in\Tc}$ \cite{Nielsen2010}, where $\Tc\subseteq \mathbb{R}$ is the set of accessible times for measurements, which may be either continuous or discrete. A map $\Phi_t$ is unital if  $\Phi_t(\ones)=\ones$. By the Kraus–Stinespring theorem \cite{Nielsen2010}, requiring $\Phi_t$ is CP is equivalent to requiring it admits an operator sum representation i.e. $\Phi_t(\cdot)=\sum_i M_{t,k} \cdot M^\dag_{t,k}$. Moreover, the expected values evolve in time with the corresponding measurement operator, leading to the coupled equations
\begin{align}\label{eq:q_sys} 
    \Sigma:=
\begin{cases}
    X_i[t]=\Phi_{t}({X_i}),\\
     y_i[t]=\Eb_\rho{(X_i[t])}, 
\end{cases}  \ \forall i.
\end{align}

The introduction of the dynamics leads to the following alternative  version of Problem \ref{prob:tom} \cite{dynamicaltomography} .\smallskip
\begin{problem}[Dynamical Quantum State Tomography]\label{prob:din_tom}
    Assuming to be able to prepare a quantum system in a given and unknown state $\rho,$ and to perform measurements of a set of observables $\Xc$ at any chosen time $t\in \Tc$, uniquely reconstruct the initial state of the system $\rho_0$ from the estimated $y_i$.
\end{problem}\smallskip

Since the state of the system is unknown, the expectation values of the observables at any time cannot be computed directly. 
The estimate $\hat{y}_i[t]$ of $y_i[t]$ can be found, as described in Appendix \ref{sec:empav}, by repeated experiments and empirical average. In particular, we let the system evolve until time $t\in \Tc$, perform the measurement of $X_i$, collect the outcome, and repeat the experiment.

This problem is particularly interesting in experimental setups where the goal is to infer information about the state in which the system is prepared, but we do not have a access to a full set of observables. In the next sections we discuss necessary and sufficient conditions on $\Xc$ and the dynamical generator in order to solve Problem \ref{prob:din_tom}.

\section{When is DQST possible?}\label{sec:dqst_feas}
\subsection{General Dynamics}

In the setup considered in this paper, we have access only to expected values of observables in 
$\Xc$ at sequences of times $t\in \Tc$. 
Therefore, a very intuitive necessary condition that ensures DQST is possible for every state in $\Dc(\Hc)$ is that  distinct states yield distinct trajectories of expectations.
Following \cite{MP-TG-FT:2024}, we introduce the following definition.\smallskip
\begin{definition}[UDDA states] \label{prb:DYNUD} A state $\rho\in \Dc(\Hc)$ is {\em uniquely dynamically determined among all states (UDDA)} if there does not exist any other state $\sigma\in \Dc(\Hc)$ such that $\Eb_\rho(X_i[t])=\Eb_\sigma(X_i[t]) \ \forall t\in \Tc, \forall X_i\in \Xc. \hfill \triangle$  
\end{definition}\smallskip
In the remainder of this section, we show that the requirement of \textit{every state being UDDA} is not only necessary for the reconstruction of (every) quantum state, but also sufficient. We will prove so by employing a well-known tool form system ad control theory: \textit{observability analysis}. 

To characterize the conditions for feasibility of DQST, we can equivalently study the set of states which are \textit{not} UDDA. In particular, two states $ {\rho}, {\sigma}\in\Dc(\Hc)$ are \textit{dynamically indistinguishable} (not UDDA) if and only if {
$$ \Eb_\rho(\Phi_t(X_i))-\Eb_\sigma(\Phi_t(X_i))=\Eb_{{\rho}-{\sigma}}(\Phi_{t}(X_i))=0 $$}
$\forall X_i\in \Xc, \forall t\in \Tc $,
where we exploited the linearity of the expectation. This is equivalent to requiring
$X_i~\in~\ker~\Eb_{{\rho}-{\sigma}}(\Phi_{t}(\cdot)) \ \forall X_i\in \Xc, \forall t\in \Tc$. 

Therefore, every state is UDDA if and only if the indistinguishability condition is not satisfied for any couple of states $\rho,\sigma \in \Dc(\Hc)$.
By leveraging  indistinguishability, we introduce the following subspaces \cite{kalman1969topics}.
\begin{definition}[Non-observable/observable subspaces]\!The~\textit{non-observable} subspace~for the system $\Sigma$ is 
\begin{align*}
\Nc&:=\{ Z \in \Bc(\Hc) \ | \ \Eb_Z(\Phi_{t}( {X_i}))=0, \forall t\in \Tc, \forall X_i\in \Xc\}.
\end{align*}
The \textit{observable} subspace is the orthogonal complement to $\Nc$ in $\Bc(\Hc)$, with respect to the standard Hilbert-Schmidt inner product:
\begin{align*}
    \Oc&:=\Nc ^\perp=span\{\Phi_{t}(X_i), \forall X_i\in \Xc, \forall t\in \Tc\}. & \quad \quad \  \triangle
\end{align*}
 \end{definition}\smallskip
It is possible to notice that a state $\rho\in \Dc(\Hc)$ is UDDA if no other state share the same trajectories of expectations and therefore, the projection of $\rho$ onto $\Nc$ is trivial (i.e. $\Pi_\Nc \rho =0$). Moreover, we highlight $\Nc$ is a subspace of traceless Hermitian operators (i.e. $\Nc\subseteq \Hf_0(\Hc)$). This follows from the fact that the identity operator $I$ always belong to the set of available measurement operators $\Xc$ and therefore to the observable subspace $\Nc^\perp=\Oc$: the inner product $\tr[Z^\dag I]$ must be equal to 0 $\forall Z \in \Nc$.

The main definition we will exploit in the reminder of the analysis is the following:
\smallskip
\begin{definition}\label{def:obs}
    The system $\Sigma$  is said to be {\em observable} if {$\Oc=\Bc(\Hc)$ (and therefore $\Nc=\{0\}$)}.
\end{definition}

The following proposition connects observability with the  requirement of every state being UDDA.
\begin{proposition}\label{prop:obs_UDDA}
    Every state $\rho\in \Dc(\Hc)$ is UDDA if and only if the system $\Sigma$ is observable.
\end{proposition}
\begin{proof}
We begin by proving that observability is a necessary condition for every state to be UDDA. By contradiction assume the system is \textit{not} observable. Given a full-rank state $\rho$ it is always possible to find $Z\in\Nc$ such that $Z\neq 0, \tr[Z]=0$, and $\epsilon > 0$ so that $\Bar{\rho}={\rho+\epsilon Z}$ is a state. This implies $\Eb_{\Bar{\rho}}(\Phi_{t}(X_i))=\Eb_{{\rho}}(\Phi_t(X_i))+\epsilon \Eb_Z(\Phi_t(X_i)))=\Eb_{\rho}(\Phi_t(X_i)) \ \forall X_i\in \Xc, \ t\in \Tc $, therefore $\rho$ and $\Bar{\rho}$ are dynamically indistinguishable and $\rho$ cannot be UDDA.

We now prove sufficiency. By contradiction suppose the system is observable and there exists two undistinguishable states $\rho,\sigma\in \Dc(\Hc),\rho\neq\sigma$. Let $Z=\rho-\sigma$, then $Z\in \Bc(\Hc), Z\neq 0$ and $\Eb_Z(\Phi_t(X_i))=0 \ \forall X_i\in \Xc, \forall t\in \Tc$. Therefore $\Nc\neq\{0\}$ and the system is not observable leading to the contradiction. 
\end{proof}
Next, we shall highlight how observability is the essential property that ensures feasibility of DQST for every state in $\Dc(\Hc)$.\smallskip 

\subsection{Linear matrix representation via vectorization}
We first introduce an alternative representation of the system $\Sigma$ in \eqref{eq:q_sys}.
A matrix $B\in \Cb^{d\times d}$ can be associated to a vector $b\in\Cb^{d^2}$ by a linear transformation (vectorization) $b={\rm vec}(B)$ that stacks the columns of $B$ one below the other so that $\vv{b}_{(j-1)d + i}=B_{ij} \ \forall \ i,j\in \{0,\dots,d-1\}.$

The key properties of this linear transformation are the following \cite{gilchrist2009vectorization}: let $A,B,C\in \mathbb{C}^{d\times d}$, then
\begin{enumerate}[P.1)]
    \item $\vect(ABC)=(C^\top \otimes A) \vect(B)$;\label{prop:P1}
    \item $\tr(A^\dag B)=\vect(A)^\dag \vect(B)$.\label{prop:P2}
\end{enumerate}
Let $\vv{x}_i\!\!=\!\!\vect(X_i)$,\! $\vv{x}_i[t]\!\!=\!\!\vect(X_i[t])$,\! $\vv{r}_0\!\!=\!\!\vect(\rho_0)$. We then have,\! 
exploiting P.\ref{prop:P1}, P.\ref{prop:P2}:
\begin{subequations}\label{eq:vec_maps}
    \begin{align}
        &\vect(\Phi_{t}(X_i))=\sum_k (M_{k,t}^\dag)^T \otimes M_{k,t} \vv{x}_i= \hat{\Phi}_{t}\vv{x}_i,\label{eq_vec_maps_a}\\
       & \vect(\Eb_{\rho_0}(X_i[t]))=\vect(\rho_0 X_i[t] )= \vv{r}_0^\dag \vv{x}_i[t]= \vv{x}_i[t]^\dag \vv{r}_0.\label{eq_vec_maps_b}
    \end{align}
\end{subequations}
Furthermore, in the rest of the paper we label with ${\rm vec}^{-1}(\cdot)$ the inverse of the map ${\rm vec}(\cdot),$ {i.e. ${\rm vec}^{-1}({\rm vec}(B))=B$ for all $B\in\Cb^{d\times d}$.} 

Thanks to the previous facts, we obtain an alternative representation of system $\Sigma$ in \eqref{eq:q_sys}:
    \begin{equation}\label{eq:dis_sys_vect_n}
    \Sigma_{v}:=\begin{cases}
     \vv{x}_i[t]=\hat{\Phi}_t\vv{x}_i,\\
        {y}_i[t]=\vv{r}_0^\dag\vv{x}_i[t]=x_i^\dag [t] r_0
    \end{cases} \forall i.
\end{equation}
This new representation enables us to formulate simple linear-algebraic conditions for observability as follows.

Let $\Rc=\{X_{i_1}[t_1],\dots X_{i_q}[t_{q}]\}$, with $\{X_{i_1},\dots X_{i_q}\}\subseteq \Xc$ and \{$t_1,\dots t_{q}\}\subseteq \Tc$, be a subset of all evolved observables. 
The vector $y_{\Rc}$ whose entries are the expectations (outputs of the system) of observables in $\Rc$ can be found by considering a matrix $O_\Rc$ whose rows are the conjugate transpose of the vectorized observables as
\begin{align}\label{eqn:r_expectations}
    y_{\Rc}=\begin{bmatrix}
        x_{i_1}^\dag[t_1] \\ \vdots \\ x^\dag_{i_q}[t_q]         
    \end{bmatrix} \vv{r}_0=O_{\Rc}\,\vv{r}_0.
\end{align}
\begin{lemma}\label{lem:obs_matrix}
    A quantum system is observable if and only if there exists a set of evolved observables $\Rc$ for which  $\rank(O_{\Rc})=d^2$.
\end{lemma}
\begin{proof} The rows of $O_{\Rc}$ 
    are given by the conjugate transpose vectorization of observables in $\Rc$. These observables are linearly independent (and their span is equal to $\Bc(\Hc)$) if and only if the corresponding vectorization are. 
    Therefore $O_\Rc$ has $d^2$ linearly independent rows (and therefore $\rank(O_{\Rc})=d^2$) if and only if there exist a set $\Rc$ of cardinlity $d^2$ such that ${\rm span}\,\Rc=\Oc=\Bc(\Hc)$, i.e. the system is observable.
\end{proof}
Notice that the matrix $O_\Rc$ satisfying the above lemma in general is not unique. In the proof of the previous proposition we discussed how $\rank(O_\Rc)=d^2$ if and only if the observables in $\Rc$ are generators of the observable subspace $\Oc$. The choice of generators of $\Oc$ in general  is not unique.

By exploiting the previous results we now establish a connection  between feasibility of DQST and observability.
\begin{proposition}\label{prop:obs_DQST}
     DQST is feasible for every state $\rho_0\in \Dc(\Hc)$ if and only if the system $\Sigma$ is observable. The vectorized state is obtained as 
    \begin{equation}\label{eqn:state_reconstruction}
         r_0=({O}_\Rc^\dag {O}_\Rc)^{-1} {O}_\Rc^\dag {y}_\Rc,
    \end{equation}
    and the state can be retrieved as $\rho_0={\rm vec}^{-1}(r_0).$
\end{proposition}
\begin{proof} We first show that observability implies DQST is feasible for every state. In view of the previous lemma, if the system is observable $\rank(O_\Rc)=d^2$ and therefore $\ker(O_\Rc)=\{0\}$. This implies that if $y_\Rc$ and $O_\Rc$ are known exactly, then the system of equations \eqref{eqn:r_expectations} admits a unique solution $r_0$. It is then possible to pre-multiply both sides of equation \eqref{eqn:r_expectations} by ${O_\Rc}^\dag$.
\begin{equation}\label{eqn:r_expectations_mod}
    O_\Rc^\dag {y}_\Rc=O_\Rc^\dag O_\Rc {r}_0.  
\end{equation}
Notice that ${O_\Rc}^\dag {O_\Rc}$ is a square invertible matrix which has the same kernel of ${O_\Rc}$. 
Therefore, \eqref{eqn:r_expectations} and \eqref{eqn:r_expectations_mod} admits the same solution and $\vv{r}_0$ can then be reconstructed exactly {by linear inversion} as \eqref{eqn:state_reconstruction}. Finally, the density operator describing the state of the system can be retrieved as ${\rho}_0=\vect^{-1}(r_0)$.

We now prove by contradiction the necessity of observability for feasibility of DQST of every state. Suppose the system is not observable, then as proven in Proposition \ref{prop:obs_UDDA} there exist two dynamically indistinguishable states $\rho_0,\sigma_0\in \Dc(\Hc)$. Let $r_0=\vect(\rho_0)$ and $s_0=\vect(\sigma_0)$, then $y_\Rc=O_\Rc r_0=O_\Rc s_0$ and \ref{eqn:r_expectations} does not admit a unique solution for all states (DQST of $\rho_0=\vect^{-1}(r_0)$ and $\sigma_0=\vect^{-1}(s_0)$ is not possible).
\end{proof}
    To conclude this subsection, we summarize the key results presented so far. The requirement that every state is UDDA is a necessary condition for DQST. As shown in Proposition \ref{prop:obs_UDDA}, this requirement is equivalent to the system being observable. Furthermore, Proposition \ref{prop:obs_DQST} establishes that observability is both necessary and sufficient for DQST of every state. Therefore, the following three conditions are equivalent: (i) the system is observable, (ii) every state is UDDA, and (iii) DQST of every state is possible. 
The highlighted equivalences between these requirements are depicted in Figure \ref{fig:equiv_cond_DQST}.
{\begin{remark}
    It is worth noting that the class of non-Markovian dynamics we consider is very general, covering virtually any (determinsitic, average) open system evolution under the hypothesis of an initially factorized state (see e.g. the standard derivation of open system dynamics in \cite{Nielsen2010}). The analysis extends non-trivially previous results: Not all dynamics can be satisfactorily {\em approximated} by Markovian ones. In particular, the reduced dynamics of a system of interest with an Hamiltonian coupling with a finite-dimensional ``environment'', due to Poincare` recurrences, exhibits strong non-Markovian behaviour, that is revival of the initial conditions on sufficiently long time scales. In discrete time, an extreme example is provided by a pair of identical systems, evolving independently with CPTP maps, whose state is swapped at the end of every time step. 
\end{remark}}

\begin{figure}
    \centering
    \includegraphics[width=0.8\linewidth]{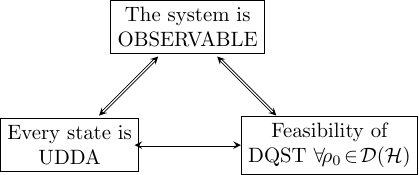}
    \caption{Equivalences between Observability, Feasibility of DQST and the requirement that every state for the system is UDDA. The three conditions are equivalent.}
    \label{fig:equiv_cond_DQST}
\end{figure}

\section{Feasibility of DQST for Markov Dynamics}
\subsection{Continuous, discrete and sampled semigroup dynamics}

We next specialize the previous to some important classes of quantum systems, whose evolution maps satisfy a forward composition law and are time-invariant. In the next sections, we show that these systems allow for a significant simplification and a systematic treatment of the DQST problem. 

A quantum system is called \textit{time-homogeneous Markovian} if it satisfies the Markov composition property 
\begin{align*}
    \Phi_{t+s} = \Phi_{t} \circ \Phi_{s}, \quad \forall t \geq s \geq 0.
\end{align*}
If such forward composition law holds, the dynamics of the system can be specified identifying a dynamical generator:
\subsubsection{Continuous Time dynamics} The evolution of measurement operators at times $\mathcal{T}= \Rb_{\geq 0}$, is described by a continuous semigroup of CP and unital maps, also known as Quantum Dynamical Semigroup (QDS) \cite{alicki2007semigroups} $\{\Phi_t\}$. The corresponding semigroup generator $\Lc$ can be expressed in Lindblad \cite{gks1979, lindblad1976} canonical form as
\begin{equation}\label{eq:lindblad}
 \Lc(\cdot)=i [H,\cdot] +\sum_k\left(L_{k}^\dag \cdot L_{k}-\frac{1}{2}\left\{L_{k}^\dag L_{k},\cdot\right\}\right).
\end{equation}
where \(H = H^\dag\) is the time-invariant Hamiltonian of the system, and the operators \(\{L_k\}\), known as noise operators, represent the non-Hamiltonian components of the generators, inducing non-unitary and irreversible dynamics.

The observables then satisfy the differential equation 
\begin{align}\label{eq:ct_sys}
    \dot{X_i}[t]=\Lc(X_i[t]).
\end{align}
Therefore, a time-homogeneous, continuous-time Markovian evolution is associated with propagators of the form $\Phi_t=e^{\mathcal{L} t}$.
\subsubsection{Discrete Time dynamics} The evolution of the measurement operators at times
$\Tc= \Nb$ is described by the difference equation
\begin{align} \label{eq:dt_sys}
    X_i[t+1]=\Phi(X_i[t]).
\end{align}
 The propagators for a time-homogeneous discrete time Markovian evolution are therefore given by the repeated application of the CP and unital map $\Phi=\Phi_1$, i.e.
\begin{align*}
    \Phi_{t}= \overset{t-times}{ \Phi \circ  \dots \circ \Phi}.
\end{align*}
\begin{remark}[\textit{Discretized evolutions}]\label{rem:discretized_ev}
In many application it is possible to perform measurements of a quantum systems of interest undergoing a continuous time evolution only at a discrete set of times $t\in \{k\Delta_t \},$ where $\Delta_t$ is a fixed sampling time interval. In this case it is often convenient to consider a discretized version of the continuous time evolution, which is described by the difference equation
\begin{align}\label{eq:discretized_sys}
    X_i[(k+1)\Delta_t]=\Phi_{\Delta t}(X_i[k\Delta t]),
\end{align}
where the propagator has the form $\Phi_{\Delta t}=e^{\Lc \Delta t}$ and $\Lc$ is the generator of the continuous time evolutions.
  \end{remark}

\subsection{Observability for Continuous-Time Semigroup  dynamics}\label{subsec:ct_dyn_obs}

Verifying observability of the system using Definition \ref{def:obs} or Lemma \ref{lem:obs_matrix} can be challenging, as it involves evaluating the trajectories of expected values of observables in $\Xc$ at times $\Tc$. While the set of available observables is finite, the set $\Tc$ may be of infinite cardinality, making the task hard.

In this and the following subsection, we focus on time-homogeneous Markovian quantum systems. Notably, for such systems the non-observable subspace can be identified by analyzing only the system’s generator, without the need to explicitly compute the propagators. This significantly simplifies the task of verifying observability.

\smallskip
\begin{proposition}\label{prop:ct_obs}
    Given a continuous-time quantum system as in \eqref{eq:ct_sys}, \eqref{eq:q_sys}, with ${\rm dim}(\Hc)=d$, there exists an integer $k^*\leq d^2-1$ such that \footnote{$\Lc^t({X_i})$ denotes the application of the map $\Lc$ to $X_i$ $t$-times.}
    \begin{align*}
\Nc =\{ Z \in \Bc(\Hc) \ | \ \Eb_Z(\Lc^t({X_i}))=0, \ \forall t\in \mathbb{N}_{\leq k^*}, \forall X_i\in \Xc\}.
\end{align*}
Accordingly, the system is observable if and only if 
\begin{equation}\label{eq:ct_obs}
    \Oc={\rm span}\{\mathcal{L}^{t}{(X_i)}, \ \ \forall X_i\in \Xc, \forall t\in \mathbb{N}_{\leq k^*} \}=\Bc(\Hc).
\end{equation}
\end{proposition}
\medskip
Proof of the proposition, based essentially on Krylov subspaces, can be found in {\cite{kalman1969topics} and \cite[Chapter 3]{wonham}}.

Similarly to Lemma \ref{lem:obs_matrix} observability of the system can be checked by evaluating the rank of a matrix that we will call \textit{continuous-time observability matrix}. 

Let $\vect{(\Lc(X_t))}=\hat L\, x_i$, where $\hat L=\hat L_H+\hat L_d$ with  
\begin{align}\label{eq:L}
    &\hat L_H=i(I\otimes H-H^\top\!\otimes I ),\\
    &\hat L_d=\sum_k\!L_k^\top\! \otimes L_k^\dag-\frac{1}{2}(I\otimes L_k^\dag L_k +(L_k^\dag L_k)^\top\!\otimes I),\nonumber
\end{align}
and let $X=\begin{bmatrix}x_1 \dots x_l \end{bmatrix}.$
The continuous-time observability matrix is defined as 
\begin{align}\label{eq:ct_obs_m}
    O_c^\dag=\begin{bmatrix}
        X & \hat L^\dag X & \dots & \hat L^  {\dag(d^2-1)} X
    \end{bmatrix}.
\end{align}
\begin{corollary} \label{cor:ct_obs}
 A continuous-time quantum system as in \eqref{eq:ct_sys}, \eqref{eq:q_sys}, with ${\rm dim}(\Hc)=d$, is observable if and only if $\rank(O_c)=d^2$.
\end{corollary}
This is the application to quantum dynamics of the {\em Kalman rank condition,} recalled in Proposition \ref{prop:krc}. The proof immediately follows from  Proposition \ref{prop:ct_obs} by following a reasoning similar to the one in the proof of Lemma \ref{lem:obs_matrix}. More in detail, it is possible to notice that the rows of $O_c$ are the vectorized generators of $\Oc$ (equation \eqref{eq:ct_obs}).
\subsection{Observability for Discrete Time Semigroup dynamics}\label{sec:dt-obs}
For discrete-time evolutions, a result similar to the one considered for continuous-time systems holds \cite{kalman1969topics, wonham}.\smallskip
\begin{proposition}\label{prop:dt_obs}
    Given a discrete-time quantum system $\Sigma$ as in \eqref{eq:dt_sys}, \eqref{eq:q_sys}, with ${\rm dim}(\Hc)=d$, there exists an integer $k^*\leq d^2-1$ such that  
\begin{align*}
\Nc =\{ Z \in \Bc(\Hc) \ | \ \Eb_Z(\Phi_t( {X_i}))=0, \ \forall t\in \mathbb{N}_{\leq k^*},\forall X_i\in \Xc\}.
\end{align*}
The system is observable if and only if 
\begin{align}\label{eq:dt_obs}
\Oc={\rm span}\{\Phi_{t}(X_i), \forall X_i\in \Xc, \forall t\in \mathbb{N}_{\leq k^*}\}=\Bc(\Hc).
\end{align} 
\end{proposition}\medskip
Similarly to the continuous time case, let $X=\begin{bmatrix}x_1 \dots x_l \end{bmatrix},$  the observability of the system can be checked by computing the rank of the \textit{discerete-time observability matrix} 
\begin{align*}
    O_d^\dag= \begin{bmatrix}
        X & \hat{\Phi}_1 X \dots \hat{\Phi}_1^{d^2-1} X  
    \end{bmatrix}.
\end{align*}
\begin{corollary}[Kalman rank condition]\label{cor:dt_obs}
   Given a discrete-time quantum system as in \eqref{eq:ct_sys}, \eqref{eq:q_sys}, with ${\rm dim}(\Hc)=d$, the system is observable if and only if $\rank(O_d)=d^2$.
\end{corollary}
As for continuous time systems the proof immediately follows from  Proposition \ref{prop:dt_obs} by following a reasoning similar to the one in the proof of Lemma \ref{lem:obs_matrix}. The rows of $O_d$ are the vectorized generators of $\Oc$ (equation \eqref{eq:dt_obs}).

For discretized evolutions, as in Remark \ref{rem:discretized_ev}, it is possible to check observability by exploiting Proposition \ref{prop:dt_obs}. However, the following lemma directly links  Proposition \ref{prop:dt_obs} and Proposition \ref{prop:ct_obs}.

\begin{lemma}\label{lem:discretized_sys_obs}
    If every couple of distinct eigenvalues $\lambda_i, \lambda_j$ of $\Lc$ for which $\Re[\lambda_i]=\Re[\lambda_j]$ are such that $\Im[\lambda_i-\lambda_j]\neq 2 \pi s/\Delta t$ $\forall s \in \Nb$ then:
    \begin{equation}
   \Oc=\{\mathcal{L}^t{(X_i)}, \ \ \forall X_i\in \Xc, \forall t\in \mathbb{N}_{\leq k^*} \} = \Bc(\Hc)
    \end{equation}
   with $k^*\leq d^2-1$  implies the discretized system is observable.
\end{lemma}
\noindent Proof of this result can be found in \cite[Theorem 3.2.1]{chen1995sampled}.

\subsection{Feasibility of DQST with a single observable}\label{sec:feas_sing_obs}

What can we say about feasibility of DQST when a single observable is availble? In this section,  we recall and extend some known feasibility results on DQST for time-homogeneous Markovian systems.

The following result has previously been established for systems evolving under discrete-time unitary dynamics, also using tools from observability analysis in systems and control theory; see \cite[Proposition 1]{XS-WY:2024}. Here we extend the result to systems undergoing continuous-time unitary dynamics.

We highlight that the framework considered in \cite{XS-WY:2024} differs slightly from ours. In particular, the vectorization of the dynamics is performed with respect to a different operator basis, and the technical assumptions on the measurement operator set $\Xc$, introduced at the end of Section \ref{sec:notions_and_assumptions}, are not imposed. As a result, the conditions required for the measurement operator set slightly  differ from those in our framework. For this reason, we revisit and adapt the proof of the latter paper to our setting. This adaptation will play a key role in understanding some of the results presented later in the paper. 

\begin{proposition}\label{prop:DQST_feas_unitaries}
    Consider a time-homogeneous Markovian quantum system \eqref{eq:q_sys} undergoing continuous time or discrete-time unitary evolutions. Let $d$ be the dimension of the associated Hilbert space.
    Then DQST is not possible when $|\Xc|<d$.
\end{proposition}
\begin{proof}
    We start by considering systems undergoing discrete-time or sampled unitary dynamics. We let $U$ be a unitary matrix with eigenvalues $\{e^{i{\phi_i}}\}_{i=1}^d$ and $\Phi(\cdot)= U \cdot U^\dag$ be the unitary evolution map. 
    For the following of the analysis it is convenient to consider the vectorized system \eqref{eq:dis_sys_vect_n}, where the vectorized evolution map is $\hat{\Phi}_1=U^*\otimes U,$ where $U^*$ is the entry-wise complex-conjugate of the matrix $U$. Moreover, we let $X^\dag$ be the matrix whose rows are the conjugate-transpose of the  vectorized measurement operators in $\Xc$.
    
    To asses observability of the vectorized system we can exploit the PBH test for observability, as presented in  \ref{sec:PBH}.
     The system is observable if and only if
     $$ {\rm rank}\left(\begin{bmatrix}
         \lambda I -\hat{\Phi}_1\\
         X^\dag
     \end{bmatrix}\right)=d^2 \quad \forall \lambda \in \mathbb{C}.$$
    
    In view of the properties of Kronecker product, it follows that $\Phi_1$ has eigenvalues which are given by the set $\{e^{-i\phi_i}e^{i\phi_j}\}_{i,j=1}^d$ and at least $d$ eigenvalues of $\Phi_1$ are 1.

    Then, for $\lambda=1$, $\rank(\lambda I-\hat{\Phi}_1)=d^2-d$. Therefore, in order to satisfy the PBH condition for observability, $X$ must have at least $d$ linearly independent rows, equivalently $\Xc$ must contain at least $d$ linearly independent operators, this leads to the claim of the proposition for discrete-time evolutions.

    We now prove the claim for continous-time unitary evolutions. 
    By contradiction suppose the system has a continuous-time dynamics, undergoes unitary evolution, $|\Xc\setminus I|<d-1$ and the system is observable.
     Let $\Lc(\cdot)=i[H,\cdot]$ be the Lindblad generator of the continuous semigroup of unitary maps describing the evolution of the system.   Then it is always possible to find a corresponding system with discretized evolution (see Remark \ref{rem:discretized_ev}) and sampling time  $\Delta t$ such that every couple of distinct eigenvalues $\lambda_i, \lambda_j$ of $\Lc$ for which $\Re[\lambda_i]=\Re[\lambda_j]$ are such that $\Im[\lambda_i-\lambda_j]\neq 2 \pi s/\Delta t$ $\forall s \in \Nb$. Then, since we are supposing the continuous time systems is observable (i.e. $\Oc=\Bc(\Hc)$), by Lemma \ref{lem:discretized_sys_obs} also the discretized system with sampling time $\Delta t$ must be observable. However this contradicts the result on observability for systems undergoing discrete-time/discretized evolutions and concludes the proof.
\end{proof}
An interesting special case is when the set of available measurement operators contains only one operator in addition to the identity. The following result has been proven in the literature  with different techniques. In particular \cite[Corollary 2]{HR-MW:2025},\cite[Proposition 1]{XS-WY:2024}, \cite{merkel2010random} consider  time-homogeneous Markovian discrete-time (or discretized) closed quantum systems. Here, we extend the result to time-homogeneous Markovian quantum systems with continuous-time dynamics. 
\begin{corollary}
    Consider a time-homogeneous Markovian quantum system \eqref{eq:q_sys} undergoing continuous-time or discrete-time unitary evolutions. Let $d>2$ be the dimension of the associated Hilbert space.
    Then DQST is not possible when $|\Xc\setminus I|=1$.
\end{corollary}
\begin{proof}
    This follows trivially from the previous proposition by considering $d>2$ and $|\Xc\setminus I|$=1.
\end{proof}

If however $d$ linearly independent observables, including the identity, are available, a generic unitary evolution  is enough for DQST of any state. This fact is proved in the following result.

\begin{proposition}
 Consider a time-homogeneous Markovian quantum system \eqref{eq:q_sys} undergoing continuous-time or discrete-time evolutions. If $|\Xc|\geq d$ then DQST is feasible for a generic unitary evolution.
\end{proposition}
\begin{proof}
Using the PBH criterion (Proposition \ref{prop:pbh}) for the vectorized system, the statement is equivalent to the following:
Let $X^\dag$ be a matrix whose rows are the conjugate transpose $d$ vectorized linearly independent operators in $\Bc(\Hc)$ (including the identity), then the set
\begin{equation}
    \Sc=\Bigl\{U\in\Uc(d) \ | \exists \lambda \ \textrm{s.t.} \ {\rm rank}\begin{bmatrix}
        (\lambda I- U^*\otimes U) \\
        {X}^\dag
    \end{bmatrix}<d^2 \Bigr\}
\end{equation}
has zero Haar measure on $\Uc(d)$.   
In the following, generic means up to a set of zero Haar measure. Let us consider a change of basis $V$ in the original space so that the chosen $U$ is diagonal, so it becomes a matrix of phases $D_U=\textrm{diag}(e^{i\phi_1},\ldots, e^{i\phi_n})=VUV^\dag$. Being $U$ generic, all the entries of the corresponding change of basis $V$ are non-zero. Being the change of basis invertible, the matrix $\tilde X$ containing the linearized version of the $d$ observables in the new basis is still full rank, and it is generically a full matrix with no zero entries, since $V$ is such and the $i$-th row of $\tilde{X}^\dag$ can be computed using the vectorization properties as  ${\rm vec}(X_i)^\dag (V^*\otimes V)^\dag,$ where $X_i$ are the $d$ linearly independent measurement operators. Now the PBH test  reads   
\begin{equation}
     \begin{bmatrix}
        (\lambda I- D_U^*\otimes D_U) \\
        \tilde X^\dag
    \end{bmatrix}.
\end{equation}
The upper block of the matrix is clearly rank-deficient, as it has $d$ eigenvalues in $1.$ It cannot however have eigenvalues that are more than $d$-degenerate, since the phases in $D_U$ are generic. For this reason,  $d$ linearly independent rows with nonzero entries of$\tilde X^\dag,$ make the full matrix full rank for all $\lambda.$ 
\end{proof}
 It is worth noting, however, that any additional constraint on the dynamics - such as locality - may cause the above to fail, since the $U$ would not be generic anymore.

\vspace{3mm} On the other hand, if we consider time-homogeneous Markovian {\em open} quantum systems with general continuous-time or discrete-time dynamics as described in the previous section, DQST is possible even when a single measurement operator is available. 

The following result has been proven in \cite[Theorem 4 and Corollary 4]{HR-MW:2025}. 

\begin{proposition} Consider time-homogeneous Markovian open quantum with continuous-time or discrete-time dynamics. If $|\Xc\setminus I|=1$ the set of maps $\Phi$ for which any state of the system can not be reconstructed via DQST is a null set.
\end{proposition}

\vspace{2mm} In the appendix we follow an alternative approach, showing that DQST for multi-qubit systems undegoing continuous-time open quantum systems is in general possible even when a single measurement operator (in addition to the identity) is available. Our result is summarized by the following proposition.

\begin{proposition} \label{prop:quibits_observability}
    Consider a system of dimension $2^N.$ There exist time-homogeneous, {\em purely} dissipative Markovian continuous-time dynamical generator for which DQST is possible when $\Xc\setminus I$ includes a single generic observable. 
\end{proposition}
Purely dissipative here denotes a Lindblad generator whose Hamiltonian can be chosen to be zero.
The proof, presented in Appendix \ref{sec:nqubitDQST}, exploits observability analysis and the structure of the pauli basis: thanks to the structure of the latter, it provides some insights on the structure of observables and dynamics that allow for DQST.

\subsection{Feasability of DQST for multipartite systems}\label{sec:multi}
We now focus on multipartite systems with  \textit{homogeneous} subsystems. Namely, the system is composed of $N$ subsystems having the same dimension $k$. The Hilbert space associated to the system is therefore $\Hc=\bigotimes_{q=1}^n \Hc_q$ where $\Hc_q=\Cb^k$. 
We assume to have access to observables that act non-trivially only on a single subsystem, i.e. $X\in \Xc$ only if $X=X_q\otimes I_{\bar{q}}$ where $X_q\in\Hf(\Hc_q)$ and $\bar{q}$ denotes all subsystems but the one indexes by $q$. We will call the latter a covering \textit{single-site} measurement.

The following corollary gives a necessary condition for observability that relies on the number of subsystems $N$ and their dimension $k$.
\begin{proposition}
    Consider an homogeneous multipartite quantum system with associated Hilbert space $\Hc=\bigotimes_{q=1}^N \Hc_q$, $\Hc_q=\Cb^k$ and all single-site measurement operators. The system, if undergoing unitary dynamics, is observable only if 
    $${k^N}\leq Nk^2-N-1.$$
\end{proposition}
\begin{proof}
    The system is observable only if $\Xc\setminus I$ includes $k^N-1$ linearly independent operators. Note that a basis for  single-site operators consists of $k^2$ elements. Thus, a basis for the available measurement operators $\Xc$ consists of $Nk^2-(N-1)$ elements. 
     The statement follows by applying the necessary condition for observability in Proposition \ref{prop:DQST_feas_unitaries}.
    \end{proof}

    \begin{remark}
        Notice that for bipartite quantum systems (N=2) the above necessary condition is always satisfied. On the other hand, for multi-qubit systems (k=2) the condition above is satisfied only up to N=3. This implies that DQST for unitary dynamics on large qubit networks is {\em never} possible using only single-site observables. This clearly highlights the importance of being able to exploit dissipative dynamics: Appendix \ref{sec:nqubitDQST} shows that a class of purely dissipative dynamics guarantee feasibility of DQST with even a single observable.
    \end{remark}

\subsection{Parametric Dynamics and Genericity of Observability}\label{sec:parametric}

In this section, we consider time-homogeneous Markov dynamics with associated generator $\Lc,$ which depends analytically on a finite number of parameters $\alpha \in \Rb^K$.   We indicate such generator as $\Lc_\alpha,$ with associated Hamiltonian and noise operators  $H_\alpha$ and  $\{L_{k,\alpha}\}$. The system selected by $\alpha$ will be denoted as $\Sigma_\alpha$, $\Sigma_{v,\alpha}$ will be its vectorized version. The non observable subspace for the system will be labeled as $\Nc_{\alpha}$.
We next recall a lemma that will be exploited to prove the main result of this section.

Consider an $m\times n$ matrix $A_\alpha = [f_{jk}(\alpha)]$, with $f_{jk} : R^K \mapsto \Cb$, such that its real and imaginary parts $\Re(f_{jk}),\Im(f_{jk})$ are (real)-analytic, and let $\mathsf{r} =    {\rm max}_{\alpha\in\Cb^K} {\rm rank}(A_\alpha)$.  
We have the following lemma \cite{ticozzi2013steadystate}:
\begin{lemma}\label{lem:matrix_gen}
    The set $\Ac=\{\alpha\in\Rb^K | \rm{rank}(A_\alpha)<\mathsf{r}\}$ is such that $\mu(\Ac)=0$, where $\mu$ is the Lebesgue measure in $\Rb^K$.
\end{lemma}
The following proposition ensures that if there exists a choice of parameters that make the system observable, then almost all of them will:
\begin{proposition} \label{prop:matrix_gen} Let $H_\alpha$, $L_{k,\alpha} \ \forall k$  be matrices such that 
each of their entries has both real and imaginary parts which are analytic functions on the parameter $\alpha\in \Rb^K.$ If $\exists \ \hat{\alpha}$ such that the system $\Sigma_{\hat{\alpha}}$ is observable, then the set $\Ac=\{\alpha \in \Rb^K \ | \ \Nc_{\alpha}\neq\emptyset \}$ is such that $\mu(\Ac)=0$.
\end{proposition}
\begin{proof}
    As in the proof of the previous Proposition 1, we can always vectorize the system $\Sigma_{\hat{\alpha}}$ to obtain
    $\Sigma_{v,\hat\alpha}$. Since the system is observable, the  observability
    matrix ${O}_{\hat{\alpha}}$ has full rank which is equal to 
    $D^2$. The imaginary and real parts of  ${O}_{\hat{\alpha}}$
    are analytic functions on the variable $\alpha$
    since they are obtained by the sum, multiplication, exponentiation of the
    entries of $H_\alpha$ and $L_{k,\alpha}$, these are all operations which
    preserve analyticity. The set $\{\alpha \in \Rb^K \ | \ {\rm rank}( {O}_{\hat{\alpha}})<D^2\}$
    corresponds exactly to the set $\Ac$ and the fact that $\mu(\Ac)=0$ follows
    from Lemma \ref{lem:matrix_gen}.
\end{proof}
This proposition then suggests that, in studying observability, we can arbitrarily set the parameter values. If the system is not observable and there exists a choice of parameters that makes it observable, by changing the values at random, we shall find a set of parameters which guarantees observability with probability one. On the other hand, if the parameters are unknown, the observability of the true dynamics is almost certainly guaranteed by that of a system with randomly chosen parameters.

An equivalent result can be derived in the same way for discrete dynamics, considering parametric Kraus maps.

\section{Reconstructing expectations of\\ unavailable observables}\label{sec:obs_exp_reconstr}
A closely related problem to state tomography is the problem of reconstructing expectation of a set of observables that we are not directly able to measure on the system, more formally:
\begin{problem}\label{prob:part_tom}
    Given a set of observables that can be measured on the system $\Xc=\{X_i\}_{i=0}^l$, a set of target observables ${\Zc}=\{Z_i\}_{i=0}^q$ that can not be directly  measured and $\rho\in \Dc(\Hc)$, predict the expectation values
    \begin{equation*}
z_i=\Eb_\rho(Z_i)=\tr({Z_i\rho}), \qquad  1 \leq i  \leq q.\end{equation*}
\end{problem}
If the system undergoes a known dynamics, it is possible to consider the problem of predicting the expectation of observables in $\Zc$ on the initial state of the system $\rho_0$, i.e. $z_i=\tr(Z_i\rho_0)$. In this section we give necessary and sufficient conditions for the feasibility of the previous problem when $\rho=\rho_0$. The conditions we will give rely on observability analysis in control theory.

If the system is observable, DQST is possible for every state and $\rho_0$ can be reconstructed. In this case, the expectation of observables in $\Zc$ can be directly computed by exploiting the reconstructed state $\rho_0$.
However, as we will show next, observability is only a sufficient but not necessary condition for the reconstruction of expectations of observables in $\Zc$. 
\begin{proposition}\label{prop:feas_obs_exp_reconstr}
    Consider the system $\Sigma$ in \eqref{eq:q_sys} and a set of target observables $\Zc=\{Z_i\}_{i=0}^q$. The expectations of observables in $\Zc$ on the initial state of the system are reconstructible if and only if
    $\Zc \subseteq \Oc$, 
    where $\Oc$ is the observable subspace, as introduced in Definition \ref{def:obs}.
\end{proposition}
\begin{proof}
We recall that in the setting considered in this paper, we can only estimate the expectation of measurement operators $\Xc$ evolved at times $\Tc$, i.e. the expectations of the set of operators $\Rc=\{X_i[t_i]\}$, where $X_i\in \Xc, t_i\in \Tc$. Let $y_i[t_i]=\Eb_{\rho_0}(X_i[t_i])$, we label as $y_\Rc=[y_1[t_1], \dots ]$ the vector of the latter expectations.
We first prove sufficiency of the condition in the statement.
If each target observable $Z_j\in \Zc$ is a linear combination of evolved measurement operators $\Rc$, i.e. 
\begin{equation}
    Z_j=\sum_i \alpha_{i,j} X_i[t_i],
\end{equation}
then its expectation $z_j$ can be found as 
\begin{equation*}
    z_j=\tr(Z_j\rho_0)=\sum_i \alpha_{i,j} \tr(X_i[t_i] \rho_0)=\sum_i \alpha_{i,j} y_i[t_i].
\end{equation*}
The above sufficient condition is equivalent to requiring 
\begin{equation*}
    Z_j \in \textrm{span}\{X_i[t_i] \ | \ X_i\in X_i, t_i\in \Tc \},
\end{equation*}
where the above linear span coincides with the observable subspace $\Oc$ defined in Section \ref{sec:dqst_feas}.

We now prove necessity by contradiction. Suppose by contradiction $\Zc \not \subseteq \Oc$ and it is possible to reconstruct the expectations of observables in $\Zc$ on the initial state. Then  there exists $Z_j\in \Zc$ which is not a linear combination of operators in $\Rc$, i.e. 
\begin{equation*}
       Z_j=\sum_i \alpha_{i,j} X_i[t_i] + \Bar{X}_j,
\end{equation*}
where $\Bar{X}_j$ is a (traceless) operator in $\Oc^\perp=\Nc$. Now we consider two initial states, a full rank state $\rho\in \Oc$ and $\sigma=(\rho+\epsilon \Bar{X_j})$, where $\epsilon>0$ is a scalar small enough such that $\sigma$ is positive semi-definite and therefore belong to $\Dc(\Hc)$. 
Then 
\begin{gather*}
    \tr(Z_j \rho)=\tr(\sum_i \alpha_{i,j} X_i[t_i] \rho) \\
   \neq
   \\
        \tr(Z_j \sigma)=\tr(\sum_i \alpha_{i,j} X_i[t_i] \rho) +\epsilon^2 \tr(\Bar{X}_j^2).
\end{gather*}
However, since $\bar{X}_j\in \Nc$ 
\begin{equation*}
    \tr(X_i[t_i] \sigma)=\tr(X_i[t_i] \rho)+\tr(X_i[t_i] \bar{X_j})=\tr(X_i[t_i] \rho)
\end{equation*}
$\forall X_i[t_i]\in \Rc$. Therefore the set of collected expectations of evolved observables always coincides for the initial states $\rho,\sigma$. This implies it impossible to distinguish the expectation of $Z_j$ for $\sigma$ and $\rho$ (which do not coincide) from estimated expectations.
\end{proof}
\begin{remark}
     The problem considered in this section is closely related to the shadow-tomography problem introduced in \cite{aaronson2018shadow, huang2020predicting}, namely the problem of predicting some properties of the system (such as observables expectations) without necessary reconstructing the full state of the system. A full comparison and a reinterpretation of shadow tomography from a dynamical viewpoint is beyond the scope of this work, and will be explored elsewhere \cite{dynamicalshadowtomography}.
\end{remark}

\section{ Quality of Dynamical Tomography via Linear Regression}\label{sec:quality}

\subsection{Mean square error analysis} In this section, we partially follow reference \cite{Qi_2013}, moreover we make the following assumption: we assume the considered quantum systems are observable.

While in the previous sections we assumed to have perfect knowledge of the expectations of observables in $\Xc$, this is never possible in practice due to the lack of knowledge on the actual system's state. 
As discussed in Appendix \ref{sec:empav}] an estimate $\hat{y}_i[t]$ of the expectation $y_i[t]$ of an observable $X_i$ at time $t$, can be obtained by averaging the outcomes of multiple measurements of $X_i[t]$ performed on identically prepared instances of the system.

If the system is observable, any initial state $\vv{r}_0$ can then be reconstructed by following Proposition \ref{prop:obs_DQST} as
\begin{align}\label{eq:state_estimate}
    \hat{\vv{r}}_0=({O}_\Rc^\dag {O}_\Rc)^{-1} {O}_\Rc^\dag \hat{y}_\Rc.
\end{align}
Notice that in the previous equation, with respect to Proposition \ref{prop:obs_DQST}, we have substituted the vector of actual expectations ${y}_\Rc$ with $\hat{y}_\Rc$, which is the vector containing the estimate of the components $y_i[t]$ of ${y}_\Rc$. Since in practice, to  estimate each expectation $y_i[t]$, we can rely on a finite number of measurements outcomes ($P$), we will always have some errors in the corresponding estimate $\hat{y}_i[t]$.
 The presence of errors in the estimates can be described by the following equation
\begin{align}
    \hat{y_i}[t]=y_i[t]+e_i[t],
\end{align}
where $e_i[t]$ is a stochastic process. The distribution of $e_i[t]$ can be deduced as follows. The estimate $\hat{y}_i$ is computed as  
\begin{align}
    \hat{y}_i[t]=\frac{v_{i,1}[t]+v_{i,2}[t]+\dots+v_{i,P}[t]}{P}
\end{align}
where the $v_{i,k}[t]$ are identically distributed random variables of mean $\tr[X_i[t] \rho]$ and finite variance $\sigma^2=\tr[X_i^2[t]\rho_0]-\tr[X_i[t]\rho_0]^2$. Each of these variables describe a single measurement on an instance of the system. 
 When $P\rightarrow\infty$, by the central limit theorem ${\hat{y}}_i$ converges in distribution to a normal with mean $\tr[X_i[t] \rho_0]=\vv{r}_0^\dag x_i[t]=y_i[t]$ and variance $\sigma_i^2[t]/P$.
This implies the measurement error $e_i[t]=\hat{y}_i[t]-y_i[t]$, converges in distribution to a normal with mean 0 and variance $\sigma_i^2[t]/P$. Since every measurement is performed on a different instance of the system, then $\mathbb{E} [e_i[t] e_i[s]]=0 \ \forall s \neq t$.

To evaluate the quality of the described DQST via {linear regression protocol}, it is possible to compute the Mean Squared Error (MSE) \cite{Qi_2013} between the actual (vectorized) state of the system $\vv{r}_0$ and its estimate $\hat{\vv{r}}_0$. 
The MSE in the described setting can be computed as
\begin{align*}
    \Eb[(\vv{r}_0\!-\!\hat{\vv{r}}_0)^\dag (\vv{r}_0\!-\!\hat{\vv{r}}_0)]\!=\!\frac{\tr[({O}_\Rc^\dag {O}_\Rc)^{-1} {O}_\Rc^\dag\Sigma {O}_\Rc({O}_\Rc^\dag {O}_\Rc)^{-1}]}{P},
\end{align*}
where $\Sigma=\, \diag(\sigma^2_0,\,\dots,\, \sigma_i^2,\, \dots) $. Note that each element of $\Sigma$ depends on $\rho_0$. However, we can derive an upper bound for the MSE which does not depend on the specific state of the system. Specifically, let  
\begin{align*}
    k=~\max_{\rho_0\in \Dc(\Hc),\, X_i\in \Xc}
\sigma_i^2,
\end{align*}
then 
\begin{equation*}
    \Eb[(\vv{r}_0-\hat{\vv{r}}_0)^\dag (\vv{r}_0-\hat{\vv{r}}_0)]\leq \frac{k}{P} \tr{[({O}_\Rc^\dag {O}_\Rc)^{-1}]}.
\end{equation*}
This upper bound on the MSE suggest that, independently on the actual state of the system, it is possible to improve the estimate of $\rho_0$ by:
\begin{enumerate}
    \item Performing a large number of experiments, i.e. by increasing the value of $P$.
    \item  By properly selecting the set of observables $\Rc$ to be measured on the system.   
\end{enumerate}

\subsection{Selection of observables and measurement times}\label{sec:alg}
According to the second point above, in order to find the minimal MSE for the protocol, we can search among all the sets of evolved observables $\Rc$ 
the sets $\Rc^*$ solving the optimization problem
\begin{align}\label{eq:opt_func} \Rc^*=\underset{\Rc}{\rm argmin} \frac{k}{P} \tr{[({O}_\Rc^\dag {O}_\Rc)^{-1}]}.
\end{align}

However, minimization of the functional in \eqref{eq:opt_func} is impractical, as we need to be able to evaluate every evolved observable corresponding to $\Xc$ at all time instants $\Tc$ and to consider in $\Rc$ every possible combination of such observables. Similar problems arise in trying to optimize standard tomographic protocols; see \cite{liang2023optimizing}. Furthermore, notice that while $\Xc$ is a set of finite cardinality, $\Tc$ may have infinite cardinality. 
In view of the difficulty of solving \eqref{eq:opt_func}, we propose an heuristic iterative procedure for determining the observables to be considered in $\Rc$.

The identification of a set $\Rc$ of small cardinality is fundamental to reduce the number of experiment and the resources required for state reconstruction. The idea behind the procedure  is therefore to maximize the information that can be acquired with a single measurement of an observable in $\Rc,$ {\em given that a number of previous measurements have already been chosen or performed}. This allows for an iterative optimization procedure, which might be suboptimal in general. 

In the following, we label with $\Rc_k$ the set of measurement operators found in steps $1,\dots, k$ and with $X^k$ the observable that is chosen in the $k$-th step of the procedure. The procedure involve the following steps.

\smallskip{\em Algorithm for Observable and Time selection (AOT):}

\smallskip \noindent In the \textit{1-st step}, without prior knowledge on which are the most informative observables, we select $X^1 \in \Xc$ at random and set $\Rc_1=X^1$.

\smallskip \noindent In the \textit{$k$-th step}, with $k>1$, we select the observable $X^k=X_i^*[t^*]$, with $X_i^*\in \Xc$, $t^*\in \Tc$, with maximum projection on $\textrm{span}\{\Rc_{k-1}\}^\perp$. More formally, let $\Pi_{\Rc^\perp}$ be the orthogonal projector onto $\textrm{span}\{\Rc_{k-1}\}^\perp$, for all $X_i\in \Xc$ we first find 
\begin{align}\label{eq:time_cost_fcn}
    t^*_i=\underset{t\in\Tc}{\textrm{argmax}} \norm{\Pi_{\Rc^\perp} X_i[t]}^2_{HS},
\end{align}
where $\norm{\cdot}_{HS}$ is the Hilbert-Schmidt norm. Successively we set
\begin{align}\label{eq:ops_cost_fcn}
X^k=\underset{X_i[t_i^*]}{\textrm{argmax}} \norm{\Pi_{\Rc^\perp} X_i[t_i^*]}^2_{HS}, 
\end{align}
and $\Rc_k=\Rc_{k-1}\cup X^k$.
These steps select the observable and the measurement time that maximizes the amount of ``new'' information, that is, orthogonal with respect to the already available one.

Finally, we \textit{stop} the procedure when $k=d^2-1$. 

\smallskip It is worth noting that, in most applications, it is not possible to let the system evolve for long time periods before performing measurements. It is therefore meaningful to restrict the search space for time instants to the set  $\Tc_{\leq T}=\{t\in\Tc | t\leq T\},$ further reducing the resources required to compute the solution to the observable selection problem. 

Although this procedure does not lead, in general, to the optimal solution of the minimization problem described above, it is a viable proxy in many practical situations. In Appendix \ref{sec:numerical_ev_alg} we provide numerical evidence that our heuristic algorithm yields good approximations to the optimal solution, while
a discussion of its numerical complexity is included in Appendix \ref{sec:computational_considerations}.

\begin{remark}[Time versus observables tradeoffs]
    We conclude the section by highlighting an emerging trade-off between the number of available, independent observables and the time needed to complete the estimation. 
    If $\Xc$ is a set of informationally complete observables, the state of the system can be reconstructed without the need to exploit the system's dynamics: all $d^2-1$ observables can be measured at time $t=0$. 
    However, we may decide to avoid using some elements of $\Xc,$ maybe because their measurements are hard, costly or lengthy to implement. Then, one or more observables needs to be measured at multiple times, at least one of which will be  greater than zero\footnote{While in principle it may be made arbitrarily small for continuous dynamics, its signal to noise ratio will be very poor - a lot of measurements would be needed to have a good estimation.}. This measurement has also to be repeated multiple times to obtain a reliable estimation of the expectation (see Appendix \ref{sec:empav}).
   Therefore,  a trade-off emerges between the number of observables in the set $\Xc$, the quality of the estimation, and the total time required to collect measurement outcomes from the physical experiment. 
\end{remark}

\section{Example I -- A chain of 4 spins}

\begin{figure}[th]
    \centering
\includegraphics[width=0.9\linewidth]{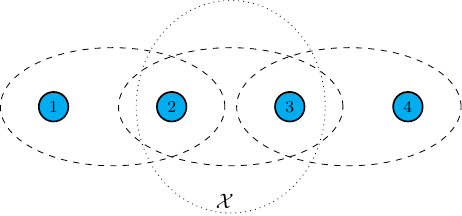}
\caption{Example 1: The 4-spin chain considered in Example I. The dotted line indicate the subsystems on which we are allowed to perform arbitrary measurements, contained in $\Xc$.
The dashed lines group the subsystems that are directly interacting dynamically.}
\label{fig:neigh_structure}
\end{figure}

\subsection{Problem setting}
In this section, we employ observability analysis to study the feasibility of DQST for some continuous-time Markovian quantum systems. {Different cases will provide insights on different aspects of the problem.}

As first example, we consider a multipartite quantum system composed of 4 qubits disposed on a line, the Hilbert space associate to the system is $\Hc=\bigotimes_{q=1}^4\Hc_q, \ \Hc_q=\Cb^2, \ \Hc \simeq \Cb^{16}$. 
Each spin interacts only with his nearest-neighbors in the line. We further assume it is possible to perform measurements on the joint system composed by the second and third spin, i.e. $\Xc = \{\sigma_2^u\sigma_3^q,\,\forall u,q=0,x,y,z\}$ where the notation $\sigma_j^\alpha$ denotes the Pauli operator $\sigma^\alpha, \alpha \in \{0,x,y,z\}$ acting on the j-th qubit, that is, $\sigma_j^\alpha \equiv I \otimes \cdots \sigma^\alpha \otimes \cdots I$ (similarly, in the following we will denote as $\sigma_j^+$ and $\sigma_j^-$ the raising and lowering operators acting on the  j-th qubit). The system is depicted in Figure \ref{fig:neigh_structure}.  
\subsection{Observability analysis}
Different evolutions of the system of interest are considered, and an observability analysis is carried out numerically for each scenario. The software \texttt{Matlab} is employed to compute the matrices $L$ and $X$ associated, respectively, with the generator of the dynamics and output maps, as defined in Section \ref{subsec:ct_dyn_obs}. Subsequently, the continuous-time observability matrix \eqref{eq:ct_obs_m} is constructed, and the Kalman rank condition was applied to determine whether the system is observable. We recall that
DQST is feasible if and only if the system is observable, 
as proved in Proposition~\ref{prop:obs_DQST}.

{\it Hamiltonian Dynamics:}  We begin by considering a purely unitary evolution, in which the system's Hamiltonian captures the interactions between neighboring spins and is given by
\begin{equation}\label{eq:sys_ham}
    H=\sum_{i=1}^4 \alpha_i\sigma_i^x + \beta_i \sigma_i^y+ \gamma_i \sigma_i^z + \sum_{i=1}^3 \delta_i \sigma_i^x\sigma_{i+1}^x+\epsilon_i \sigma_i^z\sigma_{i+1}^z,
\end{equation}
where the coefficients $\alpha_i,\beta_i,\gamma_i, \delta_i, \epsilon_i \in \Rb$ can be chosen in different ways.
When all the coefficients are set to 1, the system is not observable; in particular, the non-observable subspace has dimension 10. According to the results of Proposition~\ref{prop:matrix_gen}
, drawing \emph{all} parameters at random allows us to investigate whether the lack of observability is generic for the considered system. To this end, we performed observability analysis on 100 systems, with all parameters independently sampled from a Gaussian distribution with mean 0 and variance 1. None of these systems resulted to be observable. This indicates that the lack of observability does not depend only on the fact we set all the parameters equal to 1 but is an intrinsic property of the considered family of Hamiltonians.

{\it Dissipative Dynamics:} We consider now a second possible dynamics for the system which encompasses a dissipative term. We choose a single, local Lindblad generator with noise operator
$$L=4 \sigma_4^{-}$$
The Hamiltonian part of the dynamics is the same as before, with all the coefficients set to 1.
In contrast to the previous scenario the system is observable. Therefore, the addition of noise terms is beneficial for observability, as expected in light of the results in Section 
\ref{sec:feas_sing_obs}.
\subsection{Measurement time optimization}
Since the considered system is continuous-time Markovian, we can run the algorithm AOT described in Section \ref{sec:quality} to find the pairs of measurement operators in $\Xc$ and times in $\Tc$ identifying a set $\Rc$ of operators that can be measured on the system to solve DQST. The AOT is run only for the system in Scenario 2 as observability, is a requirement for the AOT. 
The AOT was implemented in \texttt{Matlab}, in particular the optimization of the functional in \eqref{eq:time_cost_fcn} was carried out using the \texttt{fmincon} matlab routine with initial time parameter equal to $1e-6$, moreover the search space for the optimizer was constrained to the positive real line.
In contrast, the search space for the optimization problem in  \eqref{eq:ops_cost_fcn} is a discrete set of relatively small cardinality, therefore the problem was solved in a combinatorial way by comparing the cost function for all evolved operators.

In Figure \ref{fig:time_opt} are depicted the couples of measurement operators indexes and evolution times found by the AOT. We observe that all the measurement operators in the set $\Xc$ are chosen multiple times. Only the operator with index 1 is chosen a single time, this index however correspond to the identity operator which is a fixed point for the dynamics. We observe that collecting a single outcome of an experiment to perform DQST do not require waiting large time as the maximum found time is $t_{max}\approx 12$. It is interesting to compare this with the slowest decay time of the dissipative dynamics, that is {$\tau \approx 1/|Re{(\lambda_2)}|$}, where $\lambda_2$ is the non-zero eigenvalue of the Lindblad operator with largest real part. For the considered system $\lambda_2=-0.1308$ and $\tau\approx 7,64$. Hence, the largest time selected by the AOT is under $2\tau.$

To test the proposed method for DQST, we consider three different initial states:
\begin{itemize}
    \item A pure separable state $\rho_S=\ket{0 0 0 0}\bra{0 0 0 0}$;
    \item A pure entangled state (GHZ state) $\rho_{GHZ}=\ket{\Psi}\bra{\Psi}$ with $\ket{\Psi}=(\ket{0 0 0 0}+\ket{1111})/\sqrt{2}$;
    \item A thermal Gibbs state $\rho_{Gibbs}=e^{-\beta{H}}/\tr(e^{-\beta H})$ where $H$ is the Hamiltonian of the system defined in equation \eqref{eq:sys_ham} and we chose $\beta=1$.
\end{itemize}
Let $\rho$ be equal to $\rho_S, \rho_{GHZ}$ or $\rho_{Gibbs}$, as described in section \ref{sec:quality},  as $N\rightarrow\infty$, by the central limit theorem the estimate ${\hat{y}}_i$ of expectation values of observable $R_i\in \Rc$ converges in distribution to a normal with mean  $y_i=\tr[R_i \rho]$ and variance $\sigma_i/N=(\tr[X_i^2\rho]-\tr[X_i\rho]^2)/N$. Therefore, to generate each of the estimates $\hat{y}_i$ we drew a sample from the corresponding distribution $\Nc(y_i,\sigma^2_i/N)$. Let $\hat{y}_{\Rc}=[y_1 \ y_2 \ \dots]$, we finally computed the estimate $\hat{\rho}$ of $\rho$ by following the expression in {equation \eqref{eq:state_estimate}}. 
We repeated the test for different values of $N$ and computed the squared error $\varepsilon_\rho^2=\tr({(\rho-\hat{\rho})^\dag(\rho-\hat{\rho})})$ describing the accuracy of the estimate. The results of the experiment are depicted in Figure \ref{fig:sq_error_scaling_state}.  $\varepsilon^2_\rho$ decreases linearly with the number $N$ of collected measurement outcomes for each of the considered initial states.

\begin{figure}
\centering
\includegraphics[width=\linewidth]{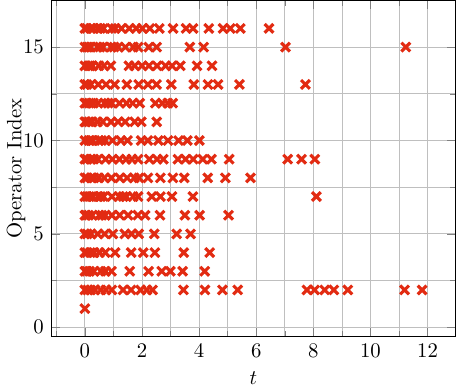}
\caption{\textit{Example 1 Scenario 2 Dissipative Dynamics:} Pairs of times and measurement operators to perform DQST identified by the AOT described in section \ref{sec:quality} {\cite{data_av}}.}
\label{fig:time_opt}
\end{figure}

\begin{figure}
\centering
\includegraphics[width=\linewidth]{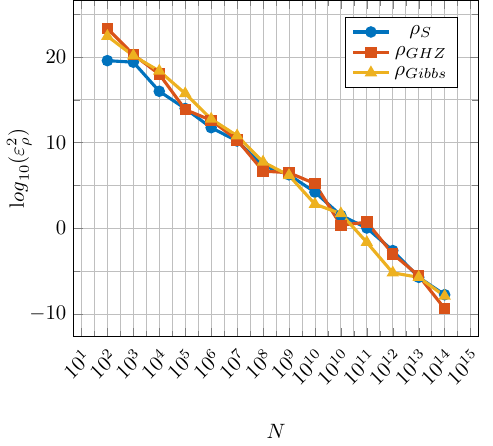}
\caption{Squared error ($\varepsilon^2_\rho$) scaling in function of the number $N$ of instances of the system involved in computing the estimate $\hat{\rho}$ of the state $\rho$ of the system. We considered three states $\rho_S, \rho_{GKS}$ and $\rho_{Gibbs}$ as described in the main text {\cite{data_av}}.}
\label{fig:sq_error_scaling_state}
\end{figure}

\section{Example II -- Coupled electron-nuclear systems}

We consider a bipartite quantum system consisting of nuclear and electronic degrees of freedom, motivated by Nitrogen-Vacancy (NV) defect center in diamond~\cite{Jacques09,enhancedNV,Jiang09,Neumann10b}. Although both electronic and nuclear spins (for the $^{14} N$ isotopes) are spin-1 systems with three levels, we consider a simplified description. Reduced models like the one considered here are commonly used when external control fields can manipulate only interactions among two of the three available levels. In the reduced model, we assume that both the nuclear and electronic degrees of freedom are represented as spin-1/2 particles therefore $\Hc=\bigotimes_{q=1}^2\Hc_q, \ \Hc_1=\Cb^4,\ \Hc_2=\Cb^2, \ \Hc \simeq \Cb^{8}$. The~system~is~depicted~in~Figure~\ref{fig:nv}.

\begin{figure}[h!]
    \centering
\includegraphics[width=0.5\linewidth]{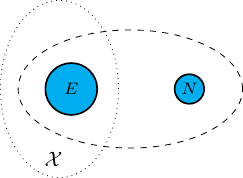}
\caption{Example 2: coupled electron (E) nuclear (N) system. The dotted line indicates the subsystem on which the measurement operators act nontrivially $\Xc$.
}\label{fig:nv}
\end{figure}

A basis for the system's state space is given
by the eight states
$$\ket{E_{el},s_{el}}\otimes\ket{s_N} \equiv
|E_{el},s_{el},s_N\rangle.$$
In the equation above, $\ket{E_{el}, s_{el}}$ represents the electronic degrees of freedom, characterized by the energy level $E_{el} = g, e$ (ground and excited states) and the electron spin $s_{el} = 0, 1$ (corresponding to spin-up and spin-down, respectively). The nuclear spin is denoted by $\ket{s_N}$ and can take the values $s_N = 0, 1$.

We assume that the system is subject to a continuous-time Markovian evolution. The electronic state can undergo a spin-preserving transition from its ground state to an excited state via optical pumping. Furthermore, the electron and the nuclei interact with each other via an Hamiltonian coupling. 
We describe the overall optically-pumped dynamics of the NV system by constructing a QDS generator as follows. Let $H_e$ and $H_g$ be the excited-state and ground-state Hamiltonian which share the same structure, the total Hamiltonian of the system is $H_{tot}=H_g+H_e$ where:
\begin{align}\label{eq:nv_hamiltonian}
    H_{g,e}&\!=\!D_{g,e} S_z^2\otimes \ones_N\!+Q\ \ones_{el}\otimes
S_z^2\nonumber \\ 
&+ B\,( g_{el} S_z \otimes \ones_N+g_n
\ones_{el}\otimes S_z )\\
&+\frac{A_{g,e}}2 ( S_x\otimes S_x +S_y
\otimes S_y+2S_z \otimes S_z )\nonumber.
\end{align}
In the above equation $S_{x,y}=\sigma_{x,y}$ are Pauli matrices on the relevant subspace, $S_z=\frac{1}{2}(\ones-\sigma_z)$ and $D_{g,e}, A_{g,e}, Q$ are fixed parameters. In particular $A_{g,e}$ determines the strength of the Hamiltonian interaction between electronic and nuclear degrees of freedom. $B$ is the intensity of the static magnetic field along the $z$-axis.
We describe the non-Hamiltonian components of the evolution with Lindblad terms consisting of jump-type operators and associated pumping and decay rates. The relevant transitions are represented by the operators below; all of them act trivially on the nuclear degrees of freedom:
\begin{align*}
    L_{1}&\!=\!\sqrt{\gamma_d}\;\ket{g,0}\bra{e,0}\otimes \ones_N, 
 & L_{2}&\!=\! \sqrt{\gamma_d}\;\ket{g,1}\bra{e,1}\otimes \ones_N,\\
L_{3}&\!=\!\sqrt{\gamma_p} \; \ket{e,0}\bra{g,0}\otimes \ones_N, & L_{4}&\!=\!\sqrt{\gamma_p} \; \ket{e,1}\bra{g,1}\otimes \ones_N.
\end{align*}
The first two operators describe decays with rates $\gamma_d,\gamma_m$, the last two operators accounts for the optical-pumping action on the electron with rate $\gamma_p$.
 Typical values of the parameters of the dynamics for NV-centers that will be employed in the following of the section are: $D_e=1420$ MHz, $D_g=2870$MHz,
    $Q=4.945$MHz, $A_e=40$ MHz, $A_g=2.2$MHz and $g_{el}=2.8$MHz/G,
    $g_n=3.08\times10^{-4}$MHz/G. The values considered for the decay rate and the optical-pumping rate are, respectively,  $\gamma_{d}=77
    \textrm{MHz}$ and $\gamma_{p}=70 \textrm{MHz}.$

{Motivated by the typical experimental setups,} we further assume to be able to perform \textit{only} measurements on the electronic degrees of freedom. In particular, we consider to measure the electron spin along the $z$-axis, i.e. $\Xc=\{\ones \otimes \sigma_z \otimes \ones, \ones \}$\footnote{We recall the identity operator is always considered in the measurement operator set $\Xc$ as described in section \ref{sec:notions_and_assumptions}}. 

We would like to assess whether DQST is feasible for the considered system, to do so, as discussed in section \ref{sec:dqst_feas}, it is possible to employ observability analysis. As for Example 1, the software \texttt{Matlab} was employed to compute the matrices $L$ and $X$ associated, respectively, with the generator of the dynamics and output maps, as defined in Section \ref{subsec:ct_dyn_obs}. Subsequently, the continuous-time observability matrix \eqref{eq:ct_obs_m} was constructed and the Kalman rank condition (Corollary \ref{cor:ct_obs}) was applied to determine whether the system is observable. 

With the latter procedure, we verified that the described system is \textit{not} observable, in particular, the observable subspace has dimension $8<d^2=64.$
Therefore, by Proposition \ref{prop:obs_DQST} DQST is not possible and the (initial) state of the system can not be uniquely reconstructed.

Motivated by the results presented in Section \ref{sec:obs_exp_reconstr}, we investigate whether it is at least feasible to reconstruct the expectation values of certain observables of interest that are \textit{not} directly measurable on the system.

We aim to reconstruct the expectation of $\sigma_z$ on the nuclei, so the considered target set is set to $\Zc=\{Z\}=\{\ones \otimes \sigma_z\}$. Notice that $\{ \Xc\cap\Zc \}=\emptyset$, which implies that the target observable cannot be measured directly on the system. By Proposition \ref{prop:feas_obs_exp_reconstr} the expectation of target observables can be reconstructed if and only if $\Zc\subseteq\Oc$. Let $s_z=\vect(Z)$, the latter condition can be checked by verifying ${\rm rank}[O_c^\dag | \, s_z]^\dag=\rank[O_c] $ where $[O_c^\dag | \, s_z]$ is the matrix obtained by stacking the column vector $s_z$ on the right of $O_c^\dag$.
We numerically verified that for the considered system, $\Zc\subseteq\Oc$,  therefore by performing a local measurement on the electron it is possible to retrieve information on the nuclei even if the state of the system cannot be reconstructed. 
In particular
\begin{equation*}
    Z\simeq\alpha_1 X[t_1]+\alpha_2 X[t_2] 
\end{equation*}
where $X=\ones \otimes \sigma_z \otimes \ones$ is the unique element of $\Xc$, $t_1=0$ and $t_2=50$, moreover $\alpha_1=2.0057$ and $\alpha_2=-1.0057$. We let $\hat{y}_1, \hat{y}_2$ be the estimates of the expectation values of  $X[t_1], X[t_2]$, we remark each of them is obtained by averaging over measurement outcomes on $N$ instances of the system. 
Then, as described in section \ref{sec:obs_exp_reconstr}, the estimate $\hat{z}$ of the expectation value $z$ of the target observable is computed as
\begin{align}\label{eq:hat_z}
    \hat{z}=\alpha_1 \hat{y}_1+\alpha_2 \hat{y}_2.
\end{align}
To test the effectiveness of the proposed method to estimate expectation of unknown observables, we consider three different initial states:
\begin{itemize}
    \item A pure separable state  $\rho_S=(\ones+\ones\otimes\ones\otimes \sigma_z)/d$. The nuclei is in the maximally mixed state and do not exhibit correlations with the target observable $Z$. The electron has an immediately recognizable correlation with $Z$;
    \item A pure entangled state (GHZ state) $\rho_{GHZ}=\ket{\Psi}\bra{\Psi}$ with $\ket{\Psi}=(\ket{0 0 0}+\ket{111})/\sqrt{2}$;
    \item A thermal Gibbs state $\rho_{Gibbs}=e^{-\beta{H}}/\tr(e^{-\beta H})$ where $H$ is the Hamiltonian of the system defined in equation \eqref{eq:nv_hamiltonian} and we chose $\beta=1$.
\end{itemize}
We then generated the expectation of the observables $X[t_1]$ and $X[t_2]$ with the same method employed in Example 1 and  computed $\hat{z}$ as described in \eqref{eq:hat_z}. We repeated the test for different values of $N$ and computed the squared error $\varepsilon_z^2=(z-\hat{z})^2,$ which quantifies the accuracy of the estimate. The results of the experiment are depicted in Figure \ref{fig:sq_error_scaling_state},where we see how  $\varepsilon^2_z$ decreases linearly with the number $N$ of collected measurement outcomes.

\begin{figure}
    \centering
\includegraphics[width=\linewidth]{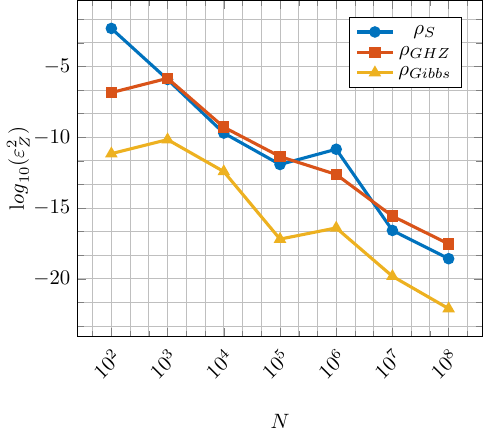}
\caption{Squared error ($\varepsilon_z^2$) scaling in function of the number $N$ of instances of the system involved in computing the estimate $\hat{z}$ of the target observable $Z$. In computing the expectations we considered three states $\rho_S, \rho_{GKS}$ and $\rho_{Gibbs}$ as described in the main text {\cite{data_av}}.}
    \label{fig:sq_error_scaling}
\end{figure}

\section{Conclusions}

In this work we lay the foundation for a systematic exploration of {\em dynamical tomography,} that is, a general method that allows for reconstructing states and expectations for a given quantum systems when limited observables are available, but we can exploit a precise knowledge of its dynamics. Leveraging control-theoretic observability analysis, we characterize systems for which any unknown state can be reconstructed from repeated experiments in the presence of general CPTP dynamics, and provide linear algebraic tests to check for such property for Markovian ones. In our framework, we also recall and extended existing results that show how unitarily evolving systems do not allow for full state reconstruction from a single observable, while open dynamics can. In presence of parametric dynamics, we show that if the system is observable for a choice of parameters, it must be so for almost all of them, allowing for randomized trials to test the feasibility of DQST. When the full state is not available, we characterize the subset of observable whose expectations can still be reconstructed from data. 
{A procedure (AOT) to select measurements and evolution times in order to maximize the amount of new information and improve the quality of the estimate is provided, which shows to significantly improve the quality of the estimate with respect to worst-case and randomized choices in the simulated scenarios.} The framework and the proposed methods are showcased in two physically-motivated examples: a chain of interacting qubits (spin 1/2), and a coupled electron-nucleus system that may serve as a simplified model for NV centers in diamonds.
In order to further develop the analysis of dynamical tomography, some additional opportunity for investigations might address, among others:
(i) a detailed error analysis and study of robustness with respect to uncertainty in the dynamical description; (ii) the development of improved algorithms for observable selection, in particular for locality constrained systems; and (iii) a in-depth comparison and connection of our approach with shadow tomography methods \cite{aaronson2018shadow}. 

\begin{acknowledgments}
F.T.acknowledges funding from the European Union - NextGenerationEU, within the National Center for HPC, Big Data and Quantum Computing (Project No.\,CN00000013, CN 1, Spoke 10) and European Union’s Horizon Europe research and innovation programme under the project “Quantum Secure Networks Partnership” (QSNP, grant agreement No 101114043). F.T and T.G. are partially supported by the Italian Ministry of University and Research under the PRIN project ``Extracting essential information and dynamics from complex networks”, grant No.\,2022MBC2EZ.
\end{acknowledgments}
\appendix
\section{Tools and Technical Results}

\subsection{Expectations from sampled averages}\label{sec:empav}
The expectation values $y_i$ are not typically accessible and thus need to be inferred form the measurement outcomes. To estimate the values of $y_i$ one has to first perform the measurement of $X_i$, collect and record the outcome, reset the experiment, and repeat this procedure multiple times, say $N$ times. Let $\hat{\alpha}_i[k]$ denote the outcome collected in the $k$-th iteration of the experiment, then the sample average estimate is give by $\hat{y}_i=\frac{1}{P} \sum_k \hat{\alpha}_i[k].$ 

Because in practice we only have access to a finite number of measurement outcomes, the estimates of the expectation values will inevitably include some error. In particular, $\hat{y}_i$ can be interpreted as the sample average of $N$ independent and identically distributed random variables associated with repeated measurements of the observable $X_i$ when the system is prepared in the same identical conditions. Each of these variables has mean $y_i$ and variance $\sigma_i^2=\tr[X_i^2\rho]-\tr[X_i\rho]^2$. By the central limit theorem we have that for $N\rightarrow\infty$ ${\hat{y}}_i$ converges in distribution to a normal with mean $y_i$ and variance $\sigma_i^2/N$.
To evaluate the quality of the described estimate it is possible to compute the Mean Squared Error (MSE) between $y_i$ and its estimate $\hat{y_i}$ which is defined as
\begin{equation}
    \mathbb{E}((y_i-\hat{y}_i)^2)=\frac{\sigma^2_i}{N}.
\end{equation}
Clearly the MSE converges linearly to 0 as $N\rightarrow\infty$ which implies $\hat{y}_i$ converges to $y_i$ almost surely.

\subsection{Control Theoretic Tools:\\ PBH and Kalman test for observability}
\label{sec:PBH}
This section quickly reviews observability analysis for linear systems, and the key tools we borrow in our work. A complete presentation can be found in e.g. \cite{marroControlledConditionedInvariants1994, wonham}.
Let us start by considering the continuous-time linear time-invariant model
\begin{equation}\label{eq:lc}
    \begin{cases}
        \dot{x}(t)=A x(t)\\
        y(t)=C x(t)
    \end{cases}\quad x_0\in\Rb^n
\end{equation}
with state $x\in\Rb^n$, output $y\in\Rb^m$, state-matrix $A\in\Rb^{n\times n}$ and output-matrix $C\in\Rc^{m\times n}$. 
Let us define the {\em non-observable subspace} as:
\[\Nc \equiv \{x\in\Rb^n |\, C e^{A t}[x_0] =0\, \forall t\geq0\}.\]
By linearity, one can note that two initial conditions $x_1,x_2\in\Rb^n$ are \textit{indistinguishable from the output}, i.e. $Ce^{At}x_1 = C e^{At} x_2$ for all $t\geq 0$, if their difference belongs to the non-observable subspace, i.e. $x_1-x_2\in\Nc$.

Similarly, if we were to consider a discrete-time linear time-invariant model 
\begin{equation}\label{eq:ld}
    \begin{cases}
        x(t+1)=A x(t)\\
        y(t)=C x(t)
    \end{cases}\quad x_0\in\Rb^n
\end{equation}
we could define the {\em non-observable subspace} as:
\[\Nc \equiv \{x\in\Rb^n |\, C A^{t}[x_0] =0\, \forall t\geq0\}.\]
Then, as in the continuous-time case we would have that two initial conditions $x_1,x_2\in\Rb^n$ are \textit{indistinguishable from the output}, i.e. $CA^{t}x_1 = C A^{t} x_2$ for all $t\geq 0$, if their difference belongs to the non-observable subspace, i.e. $x_1-x_2\in\Nc$.

In both the continuous- and discrete-time cases, 
the subspace $\mathcal{N}$ can be characterized as
\begin{equation*}
    \mathcal{N} := \ker\begin{bmatrix} C\\CA\\\vdots\\CA^{n-1}
\end{bmatrix}.
\label{eqn:non_observable_space_definition}
\end{equation*}
or, equivalently, as the largest $A$-invariant subspace contained in $\ker C$. Whenever $\Nc=\{0\}$, i.e. any two initial conditions are distinguishable from the output, we say that the linear model is \textit{observable}. An equivalent characterization is the following:

\begin{proposition}[Kalman rank condition]\label{prop:krc}
A linear model of the form \eqref{eq:lc} or \eqref{eq:ld} is \textit{observable} if and only if the matrix $\Oc \equiv \begin{bmatrix} C^T&C^TA^T&\dots&C^TA^{T n-1}
\end{bmatrix}$ has full rank, i.e. $\rank\Oc = n$.
\end{proposition}
 Another common method to assess whether or not a given linear time-invariant model is or is not observable is known as the Popov-Belevitch-Hautus (PBH) test, summarized in the next proposition. 
\begin{proposition}[PBH criterion]\label{prop:pbh}
    Let $A\in\Rb^{n\times n}$ and $C\in\Rb^{m\times n}$ be the state- and output-matrix of a linear time-invariant model (either continuous- or discrete-time). Then the model is observable if and only if $$\rank \begin{bmatrix}
        A^T-\lambda I_n &C^T
    \end{bmatrix} = n $$ for all $\lambda\in\Cb.$
\end{proposition}

\noindent Note that the criterion is automatically satisfied for all $\lambda\in\Cb$ that do not belong to the spectrum of $A$ and thus it is sufficient to check the criterion for all  $\lambda$ which belong to the spectrum of $A$.

\subsection{Reconstructing N-qubit states from a single observable and purely dissipative dynamics}\label{sec:nqubitDQST}
In this appendix, we provide a direct proof to the fact that DQST is possible with a single observable in $\Xc$ (in addition to the identity) in open quantum qubit systems.  As outlined in section \ref{sec:DQST}, for Continuous Time dynamics, the evolution of 
measurement operators is described by a continuous semigroup of CP unital maps. Let $\{F_i\}$ be an orthogonal basis for $\Bc(\Hc)$ such that $F_0=\frac{I}{\sqrt{d}}$, the semigroup generator can be expressed in the Gorini-Kossakowski-Sudarshan (GKS) form as  
\begin{equation}\label{eq:GKLS}
 \Lc(\cdot)=i [H,\cdot] +\sum_{j,i=1}^{d^2-1} a_{ji} \left( F_{i}^\dag \cdot F_{j}-\frac{1}{2}\left\{F_{j}^\dag F_{i},\cdot\right\}\right),
\end{equation}
where $H=H^\dag$ is the Hamiltonian of the system and $A=[a_{ij}]\in \mathbb{C}^{(d^2-1)\times (d^2-1) }$  is an hermitian positive semi-definite matrix (Gorini-Kosskowski-Sudarshan matrix \cite{gks1979}).\footnote{
The Lindblad canonical form of the generator in equation \ref{eq:lindblad} can be found by unitarily diagonalizing the matrix 
$A = V D V^{\dagger}$, with $D = \text{diag}(\lambda_1, \cdots, \lambda_{d^2})$, and defining $L_j = \sum_m \sqrt{\lambda_m} V_{mj} F_m$}.

A measurement operator $X_i\in \Xc$ can be expressed in the basis $\{F_k\}$ as
\begin{equation}\label{eq:vect}
    X_i=\sum_{k=0}^{d^2-1}x_{i,k}F_k,
\end{equation}
where the coefficients $x_{i,k} \in \mathbb{R}$ are given by $x_{i,k}=\tr{[X_i F_k]}$, $k\in \{0,\dots d^2-1\}$. Each operator $X_i$ can therefore be associated to a vector $x_i=[ x_{i,0},x_{i,1}  \dots ]\in \mathbb{R}^{d^2}$. \footnote{This vector representation of operators would coincide with the one considered in Sec. \ref{sec:dqst_feas} if  $\vv{F}_{(j-1)d + i-1}=\ket{i}\bra{j} \ \forall \ i,j\in \{1,\dots,d\}$, where $\ket{s}$ is the s-th vector of the canonical basis in $\mathbb{R}^{d^2}$.}

Let $x_i[0]=x_i$, the evolution of measurement operators in the considered vector representation is described by the differential equation
\begin{equation*}
    \dot{x}_{i}[t]=\Psi \,x_i[t].
\end{equation*}
The matrix $\Psi$ is given by $\Psi_{mn}=\tr[F_m \mathcal{L}(F_n)]$ and  it is possible to prove it has the following block structure
\begin{equation*}
    \Psi=\begin{bmatrix}
        0 & b^\dag\\
        0 & G^\dag
    \end{bmatrix},
\end{equation*}
where $b\in \mathbb{R}^{d^2-1}$ and $G\in\mathbb{C}^{d^2-1\times d^2-1}$.
\begin{assumption}\label{as:pauli_basis}

In the rest of the section we make the following assumptions:
\begin{enumerate}
    \item[A.1)] The dimension of the Hilbert space $\Hc$ is $d=2^N$ for some $N$.
    \item[A.2)] We choose the {\em pauli} basis $\{F_i\}$ for $\Bc(\Hc)$, namely $F_i$=$\sigma_1^i\otimes\dots\otimes \sigma_n^i$, $\forall i\in\{1,\dots d^2-1\}$ and $F_0=\frac{I}{\sqrt{d}}$. Here the $\sigma_k^i \ \forall k$ belong to the set of $2\times 2$ Pauli matrices $\{\sigma_0, \sigma_1, \sigma_2, \sigma_3\}$ with $\sigma_0=I$.
\end{enumerate}
\end{assumption}
We prove the following facts on the relationship between the structure of the semigroup generator and the matrix $\Psi$ under assumption \ref{as:pauli_basis}.

\begin{proposition}
\label{prop:diagonal_G}
    Under Assumption \ref{as:pauli_basis}, if  the GKS matrix $A$ is diagonal and the Hamiltonian $H=0$ then $\Psi$ is a diagonal matrix.
\end{proposition}
\begin{proof}
    We start the proof by recalling that for Pauli matrices $\{\sigma_1,\sigma_2,\sigma_3\}$  
    $\sigma_i \sigma_j=
        I$  if  $i=j$, and 
       $\sigma_i \sigma_j= i \epsilon_{i j k} \, \sigma_k$  if $ i \neq j$. Here $k$ is the index of the pauli matrix not considered in the l.h.s of the previous equation and $\epsilon_{i j k}$ denote the Levi-Civita symbol.

       The previous properties imply that $F_i^2=\sigma_1^i \sigma_1^i \otimes\dots\otimes \sigma_{N}^i \sigma_{N}^i =I$. Moreover, $\tr[F_i F_j]=\prod_k\tr[\sigma^i_k\sigma^j_k]=0$ if $i\neq j$ since at least one of the product $\sigma^i_k\sigma^j_k$ is equal to a traceless matrix. Finally, notice that  $\forall k \ \tr{[(\sigma^k_i\sigma^k_j)^2]}~=~\tr[(i\epsilon_{i,j,k})^2 I_2]~=~-2$.

        Since $A$ is diagonal and $H=0$
        \begin{equation*}
        b^\dag_n=\!=\!\sum_{i=1}^{d^2-1}a_{ii}\big( {\tr{\bigl[F_i F_m F_i\bigr]}}\!-\! \frac{1}{2}{\tr\bigl[\big(F_i^2 F_m +  F_m F_i^2 \big) \bigr]}\big)=0
        \end{equation*}
        where we used the cyclic property of the trace, morever
\begin{align*}
G^\dag_{nm}\!=\!\sum_{i=1}^{d^2-1}a_{ii}\big( \underbrace{\tr{\bigl[F_n F_i F_m F_i\bigr]}}_{(a)}\!-\! \frac{1}{2}\underbrace{\tr\bigl[F_n \big(F_i^2 F_m +  F_m F_i^2 \big) \bigr]}_{(b)}\big).
\end{align*}
We first analyze the term $(b)$ of the above equation, notice that 
\begin{align*}
    (b)= \begin{cases}
        2 \tr[F_n \, F_n]=2d \ \ \textrm{if} \ \ n=m,\\
        2 \tr[F_n \, F_m]=0 \ \ \textrm{if} \ \ n\neq m.
    \end{cases}
\end{align*}
Now we consider the term $(a)$, then
\begin{align*}
    (a)\!=\!\begin{cases}
              \prod_{k=1}^N \tr[(\sigma^n_k \sigma^i_k)^2]=d(-1)^{p_{n,i}}\ \ &\textrm{if} \ n=m, \\  \prod_{k=1}^N \tr[\sigma^n_k \sigma^i_k \sigma^m_k \sigma^i_k] = 0 \ &\textrm{if} \ n\neq m,
        \end{cases}
\end{align*}
where $p_{n,i}\in \mathbb{N}$ in the first equation counts the number of product $(\sigma_k^n \sigma_k^i)$ such that $\sigma^i_k \neq \sigma^n_k\neq I$ for $k\in \{1,\dots,N\}$.
By putting together the two terms we have that
\begin{equation*}    G^\dag_{nm}=\begin{cases}\sum_i a_{ii} ((-1)^{p_{n,i}}-1)d \ \ \textrm{if} \ \ n=m,\\
    0 \ \ \textrm{if} \ \ n\neq m.
    \end{cases}
\end{equation*}
\end{proof}
The previous proposition implies also that if $A$ is diagonal, $H=0$ and  $X_i\in\Hf_0(\Hc)$, then $X_i[t]\in \Hf_0(\Hc) \ \forall t$. This follows from the fact that $\dot{x}_{i,0}[t]=\dot{\tr[X_i[t]]}=0$ and therefore $x_{i,0}[t]=\tr[X_i[t]]=0\,\forall t$.
\begin{proposition}\label{prop:gen_dif_eig} 
Under assumption \ref{as:pauli_basis}, let the GKS matrix $A$ be hermitian, diagonal and positive semidefinite. Moreover, let the Hamiltonian be $H=0$. Then the eigenvalues of $G$ are generically different from each other. 
\end{proposition}
\begin{proof}
If $A$ is a diagonal matrix and $H=0$, as previously proven, $G$ is diagonal. The function mapping the vector ${a}$ of the diagonal entries of $A$ to the vector $g$ of diagonal entries of $G$ is linear in $a$.
    Let $R$ be the matrix representation of such function, i.e. $g = R\, a$, from the proof of Proposition \ref{prop:diagonal_G} it follows that its entries are given by $R_{ni}=\tr{[(F_n F_i)^2]-d}$. The rows of $R$ are all different from each other as we formally prove below. 
    
    Let $I\otimes \sigma^s_k$ be the operator $F_s=\sigma_1^s\otimes\dots\otimes \sigma_{N}^s\in \{F_i\}$ such that $\sigma_j^s=\sigma_0 \forall  j\neq k$.
    Notice that, by the properties of Pauli matrices, $R_{ns}=\tr[(F_n I\otimes \sigma^s_k)^2]-d=0$ if and only if $\sigma_k^n =\sigma^s_k$, otherwise $R_{ns}=-2d$. Therefore two rows $R_{n}, R_{m}$ of $R$ are equal if and only if $R_{ns}-R_{ms}=\tr[(F_n I\otimes \sigma^s_k)^2]-\tr[(F_m I\otimes \sigma^s_k)^2]=0 \ \forall k \in \{1,\dots, N \},\forall s$. The latter fact is true if and only if $\sigma_k^m=\sigma_k^n \ \forall k \in \{1,\dots, N \}$, i.e. $F_n=F_m$, therefore if and only if  $R_n$ and $R_m$ are the same row of $R$. 
    
    We now prove the eigenvalues of $G$ are generically different if $A$ is a generic diagonal and Hermitian matrix. We will later restrict to the case A is positive semidefinite.
    
    If $A$ is Hermitian ${a}\in\mathbb{R}^{d^{2}-1}$. We fix $n$ and $m$, and let $f({a})=g_n-g_m=(R_n-R_m) {a}$. Since ${a}\in \mathbb{R}^{d^2-1}$ and all the entries of $R$ are real, $f({a}):\mathbb{R}^{d^2-1}\rightarrow \mathbb{R}^{d^2-1}$. Moreover, all the rows of $R$ are different from each other, therefore $f({a})$ is not identically 0. This implies by \cite[Lemma 4] {ticozzi2013steadystate} the set of parameters leading to $f({a})=0$ has 0 Lebesgue measure in $\mathbb{R}^{d^2-1}$.
    The set $\Sc$ of all parameters for which 
    ${g}_n-{g}_m=0$ 
    for some $m,n$ is given by the countable union of the previously mentioned sets of measure 0. This implies $\Sc$ has also measure 0 in $\mathbb{R}^{d^2-1}$ and the eigenvalues of $G$ are generically different from each other.
    
    Notice that if we consider $\mathcal{A}=\{{{a}} \in \mathbb{R}^{(d^2-1)} | {a}_{i}\geq 0 \ \forall i\}$, then  $\Ac$ is a nonempty open subset of $\mathbb{R}^{d^2-1}$ and it has not zero Lebesgue measure in $\mathbb{R}^{d^2-1}$. This implies $\Ac \not\subset \Sc$. Moreover the set $\Sc\cap \Ac\subseteq \Sc$ has zero Lebesgue measure. This implies the set of parameters $\Sc\cap \Ac$ has zero Lebesgue measure on $\Ac$. Therefore the eigenvalues of $G$ are generically different from each other for ${a}\in \Ac$.
\end{proof}

We now exploit Proposition \ref{prop:diagonal_G} and \ref{prop:gen_dif_eig} to prove there exist continuous-time dynamics for which a quantum system is observable given a generic measurement operator $X$ in addition to the identity (i.e. $\Xc=\{I,X\}$).

We recall from Proposition \ref{prop:ct_obs}, a continuous time quantum system is observable (and therefore every state can be reconstructed via DQST) if and only if $\Oc={\rm span}\{\Lc^t(X_i)\,  \forall t\in \mathbb{N}_{\leq d^2-1}, \forall X_i\in \Xc \}=\Bc(\Hc)$. In the considered vector representation, this condition is equivalent to requiring $$ {\rm span}\{ (\Psi)^t x_i ,\ \forall t\in \mathbb{N}_{\leq d^2-1}, \forall x_i\}=\mathbb{R}^{d^2},$$
where $x_i$ is the vecotrization according to \eqref{eq:vect} of the observable $X_i\in\Xc$. 
An alternative criterion for observability is the so-called Popov-Belevitch-Hautus (PBH) test (see Appendix \ref{sec:PBH}). We define the PBH matrix as 
\begin{equation}
    P_\lambda=\begin{bmatrix}
        \lambda I-\Psi^\dag\\
        x_0^\dag\\
        \vdots\\
        x_l^\dag
    \end{bmatrix},
\end{equation}
where $x_0$ is the vectorization of the identity and therefore is a vector whose first entry is equal to 1, all other entries are 0. Moreover $l$ is the cardinality of $\Xc\setminus I$.
For the PBH criterion (see Prop. \ref{prop:pbh}), the system is observable if and only if $\rank P_\lambda=d^2 \, \forall \ \lambda\in\mathbb{C}$.

We are now ready to prove Proposition \ref{prop:quibits_observability} in the main text.
\begin{proof} {\bf [Proposition \ref{prop:quibits_observability}]}
    To prove the claim we can exploit the results of Propositions  \ref{prop:diagonal_G}, Proposition \ref{prop:gen_dif_eig} and the PBH criterion (see Prop. \ref{prop:pbh}). We consider as in proposition \ref{prop:diagonal_G} a system with purely dissipative evolution ($H=0$) and generic diagonal GKS matrix $A=A^\dag \geq 0$ in the basis $\{F_i\}$ of Assumption \ref{as:pauli_basis}. Let $X\setminus
    I =\{X\}$ and $x$ be the vectorization of $X$. The PBH matrix for the considered system is  
    \begin{equation}
    P_\lambda=\begin{bmatrix}
        \lambda &0\\
        0 & \lambda I-G\\
        1& 0\\
        0 & x^\dag_{2:d^2}
    \end{bmatrix}
    \end{equation}
    where $x_{2:d^2}$ is the sub-vector of $x$ of elements in position $i\in [2,d^2]$. 
  The matrix $G$ is diagonal and   by Proposition \ref{prop:diagonal_G} has eigenvalues that are generically different.
     The system is observable if and only if the PBH matrix $P_\lambda$ has rank $d^2 \ \forall \lambda \in \mathbb{C}$. Since $G$ is diagonal and all its eigenvalues are generically different, at most one of its rows is identically 0 $\forall \lambda \in \mathbb{C}$ (except for a zero measure set of parameters) and to guarantee that $P_\lambda$ is full rank $\forall \lambda\in \mathbb{C}$, it is necessary and sufficient to choose ${x_{2:d^2}}$ with all entries different from 0. In particular, if we choose a generic observable $X\in\Hf_0(\Hc),$ its associated vector representation ${x}_{2:d^2}$ will have all entries different from 0 except for a zero measure set of parameters. Therefore the system is observable except for a zero measure set of parameters of the dynamics and measurement operators.
 \end{proof}

 {\section{On the Measurement Selection Algorithm}
 \subsection{Comparison of selected measurements with the optimal sequence}\label{sec:numerical_ev_alg}

  The procedure for selecting the evolved observables proposed in Section \ref{sec:alg} does not in general lead to the optimal solution of the minimization problem \eqref{eq:opt_func}. However, we present next numerical evidence that the procedure helps to significantly reduce the value of the cost functional $$\Jc_\Rc,=\tr{[({O}_\Rc^\dag {O}_\Rc)^{-1}]}$$ with respect to the worst possible one, and obtain relatively comparable values to the optimal solution. 

In the following, we consider discretized systems with dissipative dynamics. In this setting, it is reasonable to assume that observables evaluated after sufficiently long evolution times carry negligible additional information about the initial state. Accordingly, we restrict the search over evolution times to the finite interval $[0,2T]$, where $T$ denotes the slowest decay time of the dynamics, defined as
\[T=\frac{1}{|Re(\lambda_2)|},\]
with $\lambda_2$ the non-zero eigenvalue of the Lindblad operator having the largest real part (the spectral gap). In all examples, we discretize this interval and take
\[
\Tc=\left\{\frac{2T\,k}{15}\;:\;k\in\{0,\dots,15\}\right\}.
\]
With this choice, it is numerically feasible to compute all evolved observables $X[t]$ with $X\in\Xc$ and $t\in\Tc$, yielding a finite pool of cardinality $q=|\Xc|\,|\Tc|$. Since DQST requires at least $d^2$ evolved observables, we can then evaluate the cost functional $J_\Rc$ over all $\binom{q}{d^2}$ possible selections of $d^2$ observables from this pool and determine its minimum.

\bigskip
\noindent\textit{Example 1: single qubit.}
We consider a qubit with Hilbert space $\Hc=\Cb^2$, and a Lindblad generator with Hamiltonian and single noise operator
\[
H=\sigma_x, \qquad L=\sigma_-,
\]
where $\sigma_x$ is the Pauli $x$ operator and $\sigma_-$ is the lowering operator.
We assume the available measurement set
\[
\Xc=\{\ket{0}\bra{0},\;\ket{+}\bra{+}\},
\]
where $\ket{0}=[1\ \ 0]^\dag$ and $\ket{+}=\frac{1}{\sqrt{2}}([1\ \ 0]^\dag+[0\ \ 1]^\dag)$. The identity need not be included since the dynamics is trace preserving.

We used \texttt{Matlab} to build the sampled observability matrix and verify the Kalman rank condition (observability is required for DQST and hence for the applicability of the selection procedure). We then:
(i) ran our heuristic selection to obtain a set $\Rc$ and its cost $J_\Rc^A$,
(ii) computed the global minimum and maximum costs, $J_{\Rc,{\rm min}}^C$ and $J_{\Rc,{\rm max}}^C$, via exhaustive search over all $d^2=4$-subsets from the $q=32$ evolved observables,
and (iii) computed the average randomized cost $J_\Rc^S$ over the retained randomized realizations.

The resulting values are:
\begin{align*}
    J_{\Rc,{\rm min}}^C&=7.6\times 10^{-1}, &
    J_\Rc^A&=1.4\times 10^1,\\
    J_{\Rc,{\rm max}}^C&=1.0\times 10^8, &
    J_{\Rc}^S&=7.2\times 10^4.
\end{align*}

These results confirm that the heuristic is sub-optimal (here $J_\Rc^A$ is about two orders of magnitude above $J_{\Rc,{\rm min}}^C$), but it substantially improves over the worst-case selection and over the randomized benchmark. In particular, the relative gap with respect to the worst-case:
\(
\frac{J_{\Rc}^A- J_{\Rc,{\rm min}}^C}{J_{\Rc,{max}}^C}=1.3\times 10^{-7},
\)
is very small, indicating that the algorithm yields a cost much closer to the best case than to the worst case on this instance. Moreover, $J_\Rc^A\ll J_\Rc^S$, showing that the heuristic improves (on average) over the simple randomized iterative procedure.

\bigskip
\noindent\textit{Example 2: two qubits.}
We consider two qubits, $\Hc=\Hc_1\otimes\Hc_2$ with $\Hc_q=\Cb^2$ and $\Hc\simeq \Cb^4$. The considered Lindblad generator is associated to an Hamiltonian:
\begin{equation}
    H=\sum_{i=1}^{2}\big(\sigma_i^x+\sigma_i^y+\sigma_i^z\big)\;+\;\sigma_1^x\sigma_{2}^x+\sigma_1^z\sigma_{2}^z,
\end{equation}
and noise operators
\begin{align*}
    L_i &= \frac{1}{\sqrt{2}}\sigma_i^+, 
    &L_{2+i} &= \frac{1}{\sqrt{2}}\sigma_i^-,  & i\in\{1,2\}.
\end{align*}
Here $\sigma_j^\alpha$ denotes the Pauli operator $\sigma^\alpha$, $\alpha \in \{0,x,y,z\}$, acting on the $j$-th qubit (e.g.\ $\sigma_1^\alpha=\sigma^\alpha\otimes I$, $\sigma_2^\alpha=I\otimes\sigma^\alpha$), and $\sigma_j^\pm$ are the raising/lowering operators on the $j$-th qubit.

We assume we can measure only the first qubit, with
\[
\Xc=\{\ket{0}\bra{0}\otimes I,\;\ket{+}\bra{+}\otimes I\}.
\]
As in Example~1, we verified observability and computed: the heuristic cost $J_\Rc^A$; and the exhaustive-search minimum/maximum costs $J_{\Rc,{\rm min}}^C$, $J_{\Rc,{\rm max}}^C$ over all $d^2=16$-subsets from the $q=32$ evolved observables that yield a full-rank observability matrix.
The obtained values are:
\begin{align*}
     J_{\Rc,{\rm min}}^C&=7.8\times 10^{5},&
     J_\Rc^A&=3.7\times 10^6,\\ 
     J_{\Rc,{\rm max}}^C&=8.6\times 10^{18},&
     J_\Rc^S&=1.1\times 10^{11}
\end{align*}

The interpretation is analogous to Example~1: the heuristic is sub-optimal (here $J_\Rc^A$ is one order of magnitude above $J_{\Rc,{\rm min}}^C$), but it yields a dramatic improvement compared with the worst-case selection and averaged cost. The relative gap with respect to the worst-case:
\(\frac{J_{\Rc}^A- J_{\Rc,{\rm min}}^C}{J_{\Rc,{max}}^C}=3.4\times 10^{-13},\)
is again very small.

Overall, these results support the claim that the heuristic algorithm---while less demanding than exhaustive search---yields good approximations to the optimal solution in the tested instances.}

{ \subsection{Computational considerations}
\label{sec:computational_considerations}

 The procedure we propose is iterative and has the following peculiarities:
    \begin{enumerate}
        \item It involves $n^2$ steps, where $n=\dim(\Hc)$.
        \item In the first step one observable at random is chosen, this operation has cost $O(1)$
        \item At every step \textit{k}, with $k>1$, it requires to search the observable which is the most orthogonal to the one selected at the previous time step. More formally, let $\Rc_k$ be the set of measurement operators found in steps $1,\dots, k$ and $X^k\in \Xc$ the observable that is chosen in the $k$-th step of the procedure. In the $k$-th step  we select the observable $X^k=X_i^*[t^*]$, with $X_i^*\in \Xc$, $t^*\in \Tc$, with maximum projection on $\textrm{span}\{\Rc_{k-1}\}^\perp$.
        
        Let $\Pi_{\Rc^\perp}$ be the orthogonal projector onto $\textrm{span}\{\Rc_{k-1}\}^\perp$, for all $X_i\in \Xc$ we first find 
\begin{align}
    t^*_i=\underset{t\in\Tc}{\textrm{argmax}} \norm{\Pi_{\Rc^\perp} X_i[t]}^2_{HS},
\end{align}
where $\norm{\cdot}_{HS}$ is the Hilbert-Schmidt norm. 

The cost of this operation is dominated by the computational cost of the algorithm chosen to perform the optimization of the functional: call this cost $A$. In practice the vectorized representation of the system can be used, therefore the evaluation of the functional requires only the multiplication of matrices (in $\Cb^{n^2\times n^2}$) and vectors (in $\Cb^{n^2}$) and matrix exponentiation (to compute the evolved observable ant time $t$). These are all operations polynomial in $n^2$ and at most O($n^6$). We emphasize that the optimization is performed for all the $|\Xc|$ observables in $\Xc$.

Successively we set
\begin{align}
X^k=\underset{X_i[t_i^*]}{\textrm{argmax}} \norm{\Pi_{\Rc^\perp} X_i[t_i^*]}^2_{HS}, 
\end{align}
and $\Rc_k=\Rc_{k-1}\cup X^k$.
This optimization requires selecting among the $|\Xc|$ observables output of the previous time step the one that minimizes the cost functional. The value of the cost functional for each of these observables has already been computed in the previous optimization step and can be stored in memory. Since no other evaluation of the functional is needed the computational cost of the operation is $O(|\Xc|)$.
\end{enumerate}

Each of the $k$ step has therefore cost O($|\Xc|A n^6$) the overall cost of the procedure is O($|\Xc|A n^8$).
}

\bibliography{bibliography}

\begin{thebibliography}{35}%
\makeatletter
\providecommand \@ifxundefined [1]{%
 \@ifx{#1\undefined}
}%
\providecommand \@ifnum [1]{%
 \ifnum #1\expandafter \@firstoftwo
 \else \expandafter \@secondoftwo
 \fi
}%
\providecommand \@ifx [1]{%
 \ifx #1\expandafter \@firstoftwo
 \else \expandafter \@secondoftwo
 \fi
}%
\providecommand \natexlab [1]{#1}%
\providecommand \enquote  [1]{``#1''}%
\providecommand \bibnamefont  [1]{#1}%
\providecommand \bibfnamefont [1]{#1}%
\providecommand \citenamefont [1]{#1}%
\providecommand \href@noop [0]{\@secondoftwo}%
\providecommand \href [0]{\begingroup \@sanitize@url \@href}%
\providecommand \@href[1]{\@@startlink{#1}\@@href}%
\providecommand \@@href[1]{\endgroup#1\@@endlink}%
\providecommand \@sanitize@url [0]{\catcode `\\12\catcode `\$12\catcode `\&12\catcode `\#12\catcode `\^12\catcode `\_12\catcode `\%12\relax}%
\providecommand \@@startlink[1]{}%
\providecommand \@@endlink[0]{}%
\providecommand \url  [0]{\begingroup\@sanitize@url \@url }%
\providecommand \@url [1]{\endgroup\@href {#1}{\urlprefix }}%
\providecommand \urlprefix  [0]{URL }%
\providecommand \Eprint [0]{\href }%
\providecommand \doibase [0]{https://doi.org/}%
\providecommand \selectlanguage [0]{\@gobble}%
\providecommand \bibinfo  [0]{\@secondoftwo}%
\providecommand \bibfield  [0]{\@secondoftwo}%
\providecommand \translation [1]{[#1]}%
\providecommand \BibitemOpen [0]{}%
\providecommand \bibitemStop [0]{}%
\providecommand \bibitemNoStop [0]{.\EOS\space}%
\providecommand \EOS [0]{\spacefactor3000\relax}%
\providecommand \BibitemShut  [1]{\csname bibitem#1\endcsname}%
\let\auto@bib@innerbib\@empty
\bibitem [{\citenamefont {Holevo}(2001)}]{holevo}%
  \BibitemOpen
  \bibfield  {author} {\bibinfo {author} {\bibfnamefont {A.}~\bibnamefont {Holevo}},\ }\href@noop {} {\emph {\bibinfo {title} {Statistical Structure of Quantum Theory}}},\ Lecture Notes in Physics; Monographs: {\bf 67}\ (\bibinfo  {publisher} {Springer-Verlag, Berlin},\ \bibinfo {year} {2001})\BibitemShut {NoStop}%
\bibitem [{\citenamefont {Paris}\ and\ \citenamefont {Rehacek}(2004)}]{paris-estimationbook}%
  \BibitemOpen
  \bibfield  {author} {\bibinfo {author} {\bibfnamefont {M.}~\bibnamefont {Paris}}\ and\ \bibinfo {author} {\bibfnamefont {J.}~\bibnamefont {Rehacek}},\ }\href@noop {} {\emph {\bibinfo {title} {Quantum State Estimation}}},\ Lecture Notes in Physics\ (\bibinfo  {publisher} {Springer Berlin Heidelberg},\ \bibinfo {year} {2004})\BibitemShut {NoStop}%
\bibitem [{\citenamefont {Kliesch}\ and\ \citenamefont {Roth}(2021)}]{quantumverificationreview}%
  \BibitemOpen
  \bibfield  {author} {\bibinfo {author} {\bibfnamefont {M.}~\bibnamefont {Kliesch}}\ and\ \bibinfo {author} {\bibfnamefont {I.}~\bibnamefont {Roth}},\ }\bibfield  {title} {\bibinfo {title} {Theory of quantum system certification},\ }\href@noop {} {\bibfield  {journal} {\bibinfo  {journal} {PRX Quantum}\ }\textbf {\bibinfo {volume} {2}},\ \bibinfo {pages} {010201} (\bibinfo {year} {2021})}\BibitemShut {NoStop}%
\bibitem [{\citenamefont {Kech}(2016)}]{dynamicaltomography}%
  \BibitemOpen
  \bibfield  {author} {\bibinfo {author} {\bibfnamefont {M.}~\bibnamefont {Kech}},\ }\bibfield  {title} {\bibinfo {title} {Dynamical quantum tomography},\ }\href@noop {} {\bibfield  {journal} {\bibinfo  {journal} {Journal of Mathematical Physics}\ }\textbf {\bibinfo {volume} {57}} (\bibinfo {year} {2016})}\BibitemShut {NoStop}%
\bibitem [{\citenamefont {Merkel}\ \emph {et~al.}(2010)\citenamefont {Merkel}, \citenamefont {Riofrio}, \citenamefont {Flammia},\ and\ \citenamefont {Deutsch}}]{merkel2010random}%
  \BibitemOpen
  \bibfield  {author} {\bibinfo {author} {\bibfnamefont {S.~T.}\ \bibnamefont {Merkel}}, \bibinfo {author} {\bibfnamefont {C.~A.}\ \bibnamefont {Riofrio}}, \bibinfo {author} {\bibfnamefont {S.~T.}\ \bibnamefont {Flammia}},\ and\ \bibinfo {author} {\bibfnamefont {I.~H.}\ \bibnamefont {Deutsch}},\ }\bibfield  {title} {\bibinfo {title} {Random unitary maps for quantum state reconstruction},\ }\href@noop {} {\bibfield  {journal} {\bibinfo  {journal} {Physical Review A—Atomic, Molecular, and Optical Physics}\ }\textbf {\bibinfo {volume} {81}},\ \bibinfo {pages} {032126} (\bibinfo {year} {2010})}\BibitemShut {NoStop}%
\bibitem [{\citenamefont {Rall}\ and\ \citenamefont {Wolf}(2025)}]{HR-MW:2025}%
  \BibitemOpen
  \bibfield  {author} {\bibinfo {author} {\bibfnamefont {H.}~\bibnamefont {Rall}}\ and\ \bibinfo {author} {\bibfnamefont {M.~M.}\ \bibnamefont {Wolf}},\ }\bibfield  {title} {\bibinfo {title} {Quantum tomography from the evolution of a single expectation},\ }in\ \href@noop {} {\emph {\bibinfo {booktitle} {Annales Henri Poincar{\'e}}}}\ (\bibinfo {organization} {Springer},\ \bibinfo {year} {2025})\ pp.\ \bibinfo {pages} {1--24}\BibitemShut {NoStop}%
\bibitem [{\citenamefont {Peruzzo}\ \emph {et~al.}(2024)\citenamefont {Peruzzo}, \citenamefont {Grigoletto},\ and\ \citenamefont {Ticozzi}}]{MP-TG-FT:2024}%
  \BibitemOpen
  \bibfield  {author} {\bibinfo {author} {\bibfnamefont {M.}~\bibnamefont {Peruzzo}}, \bibinfo {author} {\bibfnamefont {T.}~\bibnamefont {Grigoletto}},\ and\ \bibinfo {author} {\bibfnamefont {F.}~\bibnamefont {Ticozzi}},\ }\bibfield  {title} {\bibinfo {title} {Reconstructing quantum states from local observation: A dynamical viewpoint},\ }in\ \href {https://doi.org/10.1109/CDC56724.2024.10886855} {\emph {\bibinfo {booktitle} {2024 IEEE 63rd Conference on Decision and Control (CDC)}}}\ (\bibinfo {year} {2024})\ pp.\ \bibinfo {pages} {775--780}\BibitemShut {NoStop}%
\bibitem [{\citenamefont {Xiao}\ \emph {et~al.}(2024)\citenamefont {Xiao}, \citenamefont {Wang}, \citenamefont {Yu}, \citenamefont {Zhang}, \citenamefont {Dong},\ and\ \citenamefont {Petersen}}]{XS-WY:2024}%
  \BibitemOpen
  \bibfield  {author} {\bibinfo {author} {\bibfnamefont {S.}~\bibnamefont {Xiao}}, \bibinfo {author} {\bibfnamefont {Y.}~\bibnamefont {Wang}}, \bibinfo {author} {\bibfnamefont {Q.}~\bibnamefont {Yu}}, \bibinfo {author} {\bibfnamefont {J.}~\bibnamefont {Zhang}}, \bibinfo {author} {\bibfnamefont {D.}~\bibnamefont {Dong}},\ and\ \bibinfo {author} {\bibfnamefont {I.~R.}\ \bibnamefont {Petersen}},\ }\bibfield  {title} {\bibinfo {title} {Quantum state tomography from observable time traces in closed quantum systems},\ }\href@noop {} {\bibfield  {journal} {\bibinfo  {journal} {Control Theory and Technology}\ }\textbf {\bibinfo {volume} {22}},\ \bibinfo {pages} {222} (\bibinfo {year} {2024})}\BibitemShut {NoStop}%
\bibitem [{\citenamefont {Kalman}\ \emph {et~al.}(1969)\citenamefont {Kalman}, \citenamefont {Falb},\ and\ \citenamefont {Arbib}}]{kalman1969topics}%
  \BibitemOpen
  \bibfield  {author} {\bibinfo {author} {\bibfnamefont {R.~E.}\ \bibnamefont {Kalman}}, \bibinfo {author} {\bibfnamefont {P.~L.}\ \bibnamefont {Falb}},\ and\ \bibinfo {author} {\bibfnamefont {M.~A.}\ \bibnamefont {Arbib}},\ }\href@noop {} {\emph {\bibinfo {title} {Topics in mathematical system theory}}},\ Vol.~\bibinfo {volume} {1}\ (\bibinfo  {publisher} {McGraw-Hill New York},\ \bibinfo {year} {1969})\BibitemShut {NoStop}%
\bibitem [{\citenamefont {Linden}\ \emph {et~al.}(2002)\citenamefont {Linden}, \citenamefont {Popescu},\ and\ \citenamefont {Wootters}}]{linden_almost_2002}%
  \BibitemOpen
  \bibfield  {author} {\bibinfo {author} {\bibfnamefont {N.}~\bibnamefont {Linden}}, \bibinfo {author} {\bibfnamefont {S.}~\bibnamefont {Popescu}},\ and\ \bibinfo {author} {\bibfnamefont {W.~K.}\ \bibnamefont {Wootters}},\ }\bibfield  {title} {\bibinfo {title} {Almost every pure state of three qubits is completely determined by its two-particle reduced density matrices},\ }\href@noop {} {\bibfield  {journal} {\bibinfo  {journal} {Phys. Rev. Lett.}\ }\textbf {\bibinfo {volume} {89}},\ \bibinfo {pages} {207901} (\bibinfo {year} {2002})}\BibitemShut {NoStop}%
\bibitem [{\citenamefont {Linden}\ and\ \citenamefont {Wootters}(2002)}]{Linden2002_2}%
  \BibitemOpen
  \bibfield  {author} {\bibinfo {author} {\bibfnamefont {N.}~\bibnamefont {Linden}}\ and\ \bibinfo {author} {\bibfnamefont {W.~K.}\ \bibnamefont {Wootters}},\ }\bibfield  {title} {\bibinfo {title} {The parts determine the whole in a generic pure quantum state},\ }\href@noop {} {\bibfield  {journal} {\bibinfo  {journal} {Phys. Rev. Lett.}\ }\textbf {\bibinfo {volume} {89}},\ \bibinfo {pages} {277906} (\bibinfo {year} {2002})}\BibitemShut {NoStop}%
\bibitem [{\citenamefont {Weis}\ and\ \citenamefont {Gouveia}(2023)}]{weisQuantumMarginalsFaces2023}%
  \BibitemOpen
  \bibfield  {author} {\bibinfo {author} {\bibfnamefont {S.}~\bibnamefont {Weis}}\ and\ \bibinfo {author} {\bibfnamefont {J.}~\bibnamefont {Gouveia}},\ }\bibfield  {title} {\bibinfo {title} {The face lattice of the set of reduced density matrices and its coatoms},\ }\href@noop {} {\bibfield  {journal} {\bibinfo  {journal} {Information Geometry}\ ,\ \bibinfo {pages} {1}} (\bibinfo {year} {2023})}\BibitemShut {NoStop}%
\bibitem [{\citenamefont {Karuvade}\ \emph {et~al.}(2019)\citenamefont {Karuvade}, \citenamefont {Johnson}, \citenamefont {Ticozzi},\ and\ \citenamefont {Viola}}]{karuvade2019uniquely}%
  \BibitemOpen
  \bibfield  {author} {\bibinfo {author} {\bibfnamefont {S.}~\bibnamefont {Karuvade}}, \bibinfo {author} {\bibfnamefont {P.~D.}\ \bibnamefont {Johnson}}, \bibinfo {author} {\bibfnamefont {F.}~\bibnamefont {Ticozzi}},\ and\ \bibinfo {author} {\bibfnamefont {L.}~\bibnamefont {Viola}},\ }\bibfield  {title} {\bibinfo {title} {Uniquely determined pure quantum states need not be unique ground states of quasi-local hamiltonians},\ }\href@noop {} {\bibfield  {journal} {\bibinfo  {journal} {Physical Review A}\ }\textbf {\bibinfo {volume} {99}},\ \bibinfo {pages} {062104} (\bibinfo {year} {2019})}\BibitemShut {NoStop}%
\bibitem [{\citenamefont {Xin}\ \emph {et~al.}(2017)\citenamefont {Xin}, \citenamefont {Lu}, \citenamefont {Klassen}, \citenamefont {Yu}, \citenamefont {Ji}, \citenamefont {Chen}, \citenamefont {Ma}, \citenamefont {Long}, \citenamefont {Zeng},\ and\ \citenamefont {Laflamme}}]{Xin2017}%
  \BibitemOpen
  \bibfield  {author} {\bibinfo {author} {\bibfnamefont {T.}~\bibnamefont {Xin}}, \bibinfo {author} {\bibfnamefont {D.}~\bibnamefont {Lu}}, \bibinfo {author} {\bibfnamefont {J.}~\bibnamefont {Klassen}}, \bibinfo {author} {\bibfnamefont {N.}~\bibnamefont {Yu}}, \bibinfo {author} {\bibfnamefont {Z.}~\bibnamefont {Ji}}, \bibinfo {author} {\bibfnamefont {J.}~\bibnamefont {Chen}}, \bibinfo {author} {\bibfnamefont {X.}~\bibnamefont {Ma}}, \bibinfo {author} {\bibfnamefont {G.}~\bibnamefont {Long}}, \bibinfo {author} {\bibfnamefont {B.}~\bibnamefont {Zeng}},\ and\ \bibinfo {author} {\bibfnamefont {R.}~\bibnamefont {Laflamme}},\ }\bibfield  {title} {\bibinfo {title} {Quantum state tomography via reduced density matrices},\ }\href@noop {} {\bibfield  {journal} {\bibinfo  {journal} {Phys. Rev. Lett.}\ }\textbf {\bibinfo {volume} {118}},\ \bibinfo {pages} {020401} (\bibinfo {year} {2017})}\BibitemShut {NoStop}%
\bibitem [{\citenamefont {Zorzi}\ \emph {et~al.}(2014{\natexlab{a}})\citenamefont {Zorzi}, \citenamefont {Ticozzi},\ and\ \citenamefont {Ferrante}}]{zorzi2014}%
  \BibitemOpen
  \bibfield  {author} {\bibinfo {author} {\bibfnamefont {M.}~\bibnamefont {Zorzi}}, \bibinfo {author} {\bibfnamefont {F.}~\bibnamefont {Ticozzi}},\ and\ \bibinfo {author} {\bibfnamefont {A.}~\bibnamefont {Ferrante}},\ }\bibfield  {title} {\bibinfo {title} {Minimum relative entropy for quantum estimation: Feasibility and general solution},\ }\href {https://doi.org/10.1109/TIT.2013.2286087} {\bibfield  {journal} {\bibinfo  {journal} {IEEE Transactions on Information Theory}\ }\textbf {\bibinfo {volume} {60}},\ \bibinfo {pages} {357} (\bibinfo {year} {2014}{\natexlab{a}})}\BibitemShut {NoStop}%
\bibitem [{\citenamefont {Zorzi}\ \emph {et~al.}(2014{\natexlab{b}})\citenamefont {Zorzi}, \citenamefont {Ticozzi},\ and\ \citenamefont {Ferrante}}]{zorzi2014minimal}%
  \BibitemOpen
  \bibfield  {author} {\bibinfo {author} {\bibfnamefont {M.}~\bibnamefont {Zorzi}}, \bibinfo {author} {\bibfnamefont {F.}~\bibnamefont {Ticozzi}},\ and\ \bibinfo {author} {\bibfnamefont {A.}~\bibnamefont {Ferrante}},\ }\bibfield  {title} {\bibinfo {title} {Minimal resources identifiability and estimation of quantum channels},\ }\href@noop {} {\bibfield  {journal} {\bibinfo  {journal} {Quantum information processing}\ }\textbf {\bibinfo {volume} {13}},\ \bibinfo {pages} {683} (\bibinfo {year} {2014}{\natexlab{b}})}\BibitemShut {NoStop}%
\bibitem [{\citenamefont {Nielsen}\ and\ \citenamefont {Chuang}(2010)}]{Nielsen2010}%
  \BibitemOpen
  \bibfield  {author} {\bibinfo {author} {\bibfnamefont {M.~A.}\ \bibnamefont {Nielsen}}\ and\ \bibinfo {author} {\bibfnamefont {I.~L.}\ \bibnamefont {Chuang}},\ }\href@noop {} {\emph {\bibinfo {title} {Quantum Computation and Quantum Information: 10th Anniversary Edition}}}\ (\bibinfo  {publisher} {Cambridge University Press},\ \bibinfo {year} {2010})\BibitemShut {NoStop}%
\bibitem [{\citenamefont {Gilchrist}\ \emph {et~al.}(2009)\citenamefont {Gilchrist}, \citenamefont {Terno},\ and\ \citenamefont {Wood}}]{gilchrist2009vectorization}%
  \BibitemOpen
  \bibfield  {author} {\bibinfo {author} {\bibfnamefont {A.}~\bibnamefont {Gilchrist}}, \bibinfo {author} {\bibfnamefont {D.~R.}\ \bibnamefont {Terno}},\ and\ \bibinfo {author} {\bibfnamefont {C.~J.}\ \bibnamefont {Wood}},\ }\bibfield  {title} {\bibinfo {title} {Vectorization of quantum operations and its use},\ }\href@noop {} {\bibfield  {journal} {\bibinfo  {journal} {arXiv preprint arXiv:0911.2539}\ } (\bibinfo {year} {2009})}\BibitemShut {NoStop}%
\bibitem [{\citenamefont {Alicki}\ and\ \citenamefont {Lendi}(2007)}]{alicki2007semigroups}%
  \BibitemOpen
  \bibfield  {author} {\bibinfo {author} {\bibfnamefont {R.}~\bibnamefont {Alicki}}\ and\ \bibinfo {author} {\bibfnamefont {K.}~\bibnamefont {Lendi}},\ }\href@noop {} {\emph {\bibinfo {title} {Quantum dynamical semigroups and applications}}},\ Vol.\ \bibinfo {volume} {717}\ (\bibinfo  {publisher} {Springer},\ \bibinfo {year} {2007})\BibitemShut {NoStop}%
\bibitem [{\citenamefont {Gorini}\ \emph {et~al.}(1976)\citenamefont {Gorini}, \citenamefont {Kossakowski},\ and\ \citenamefont {Sudarshan}}]{gks1979}%
  \BibitemOpen
  \bibfield  {author} {\bibinfo {author} {\bibfnamefont {V.}~\bibnamefont {Gorini}}, \bibinfo {author} {\bibfnamefont {A.}~\bibnamefont {Kossakowski}},\ and\ \bibinfo {author} {\bibfnamefont {E.~C.~G.}\ \bibnamefont {Sudarshan}},\ }\bibfield  {title} {\bibinfo {title} {{Completely positive dynamical semigroups of N‐level systems}},\ }\href {https://doi.org/10.1063/1.522979} {\bibfield  {journal} {\bibinfo  {journal} {Journal of Mathematical Physics}\ }\textbf {\bibinfo {volume} {17}},\ \bibinfo {pages} {821} (\bibinfo {year} {1976})},\ \Eprint {https://arxiv.org/abs/https://pubs.aip.org/aip/jmp/article-pdf/17/5/821/19090720/821\_1\_online.pdf} {https://pubs.aip.org/aip/jmp/article-pdf/17/5/821/19090720/821\_1\_online.pdf} \BibitemShut {NoStop}%
\bibitem [{\citenamefont {Lindblad}(1976)}]{lindblad1976}%
  \BibitemOpen
  \bibfield  {author} {\bibinfo {author} {\bibfnamefont {G.}~\bibnamefont {Lindblad}},\ }\bibfield  {title} {\bibinfo {title} {On the generators of quantum dynamical semigroups},\ }\href {https://doi.org/10.1007/BF01608499} {\bibfield  {journal} {\bibinfo  {journal} {Communications in Mathematical Physics}\ }\textbf {\bibinfo {volume} {48}},\ \bibinfo {pages} {119} (\bibinfo {year} {1976})}\BibitemShut {NoStop}%
\bibitem [{\citenamefont {Murray~Wonham}(1979)}]{wonham}%
  \BibitemOpen
  \bibfield  {author} {\bibinfo {author} {\bibfnamefont {W.}~\bibnamefont {Murray~Wonham}},\ }\href@noop {} {\emph {\bibinfo {title} {Linear Multivariable Control: A Geometric Approach}}}\ (\bibinfo  {publisher} {Springer New York},\ \bibinfo {year} {1979})\BibitemShut {NoStop}%
\bibitem [{\citenamefont {Tongwen~Chen}(1995)}]{chen1995sampled}%
  \BibitemOpen
  \bibfield  {author} {\bibinfo {author} {\bibfnamefont {B.~A.~F.}\ \bibnamefont {Tongwen~Chen}},\ }\href@noop {} {\emph {\bibinfo {title} {Optimal sampled-data control systems}}}\ (\bibinfo  {publisher} {Springer London},\ \bibinfo {year} {1995})\BibitemShut {NoStop}%
\bibitem [{\citenamefont {Ticozzi}\ and\ \citenamefont {Viola}(2014)}]{ticozzi2013steadystate}%
  \BibitemOpen
  \bibfield  {author} {\bibinfo {author} {\bibfnamefont {F.}~\bibnamefont {Ticozzi}}\ and\ \bibinfo {author} {\bibfnamefont {L.}~\bibnamefont {Viola}},\ }\bibfield  {title} {\bibinfo {title} {Steady-state entanglement by engineered quasi-local markovian dissipation: Hamiltonian-assisted and conditional stabilization},\ }\href {https://doi.org/10.26421/QIC14.3-4-5} {\bibfield  {journal} {\bibinfo  {journal} {Quantum Inf. Comput.}\ }\textbf {\bibinfo {volume} {14}},\ \bibinfo {pages} {265} (\bibinfo {year} {2014})}\BibitemShut {NoStop}%
\bibitem [{\citenamefont {Aaronson}(2018)}]{aaronson2018shadow}%
  \BibitemOpen
  \bibfield  {author} {\bibinfo {author} {\bibfnamefont {S.}~\bibnamefont {Aaronson}},\ }\bibfield  {title} {\bibinfo {title} {Shadow tomography of quantum states},\ }in\ \href@noop {} {\emph {\bibinfo {booktitle} {Proceedings of the 50th annual ACM SIGACT symposium on theory of computing}}}\ (\bibinfo {year} {2018})\ pp.\ \bibinfo {pages} {325--338}\BibitemShut {NoStop}%
\bibitem [{\citenamefont {Huang}\ \emph {et~al.}(2020)\citenamefont {Huang}, \citenamefont {Kueng},\ and\ \citenamefont {Preskill}}]{huang2020predicting}%
  \BibitemOpen
  \bibfield  {author} {\bibinfo {author} {\bibfnamefont {H.-Y.}\ \bibnamefont {Huang}}, \bibinfo {author} {\bibfnamefont {R.}~\bibnamefont {Kueng}},\ and\ \bibinfo {author} {\bibfnamefont {J.}~\bibnamefont {Preskill}},\ }\bibfield  {title} {\bibinfo {title} {Predicting many properties of a quantum system from very few measurements},\ }\href@noop {} {\bibfield  {journal} {\bibinfo  {journal} {Nature Physics}\ }\textbf {\bibinfo {volume} {16}},\ \bibinfo {pages} {1050} (\bibinfo {year} {2020})}\BibitemShut {NoStop}%
\bibitem [{\citenamefont {Peruzzo}\ and\ \citenamefont {Ticozzi}(2025)}]{dynamicalshadowtomography}%
  \BibitemOpen
  \bibfield  {author} {\bibinfo {author} {\bibfnamefont {M.}~\bibnamefont {Peruzzo}}\ and\ \bibinfo {author} {\bibfnamefont {F.}~\bibnamefont {Ticozzi}},\ }\href@noop {} {\bibinfo {title} {Dynamical shadow tomography}} (\bibinfo {year} {2025}),\ \bibinfo {note} {in preparation}\BibitemShut {NoStop}%
\bibitem [{\citenamefont {Qi}\ \emph {et~al.}(2013)\citenamefont {Qi}, \citenamefont {Hou}, \citenamefont {Li}, \citenamefont {Dong}, \citenamefont {Xiang},\ and\ \citenamefont {Guo}}]{Qi_2013}%
  \BibitemOpen
  \bibfield  {author} {\bibinfo {author} {\bibfnamefont {B.}~\bibnamefont {Qi}}, \bibinfo {author} {\bibfnamefont {Z.}~\bibnamefont {Hou}}, \bibinfo {author} {\bibfnamefont {L.}~\bibnamefont {Li}}, \bibinfo {author} {\bibfnamefont {D.}~\bibnamefont {Dong}}, \bibinfo {author} {\bibfnamefont {G.}~\bibnamefont {Xiang}},\ and\ \bibinfo {author} {\bibfnamefont {G.}~\bibnamefont {Guo}},\ }\bibfield  {title} {\bibinfo {title} {Quantum state tomography via linear regression estimation},\ }\href@noop {} {\bibfield  {journal} {\bibinfo  {journal} {Scientific Reports}\ }\textbf {\bibinfo {volume} {3}} (\bibinfo {year} {2013})}\BibitemShut {NoStop}%
\bibitem [{\citenamefont {Liang}\ \emph {et~al.}(2023)\citenamefont {Liang}, \citenamefont {Ticozzi},\ and\ \citenamefont {Vallone}}]{liang2023optimizing}%
  \BibitemOpen
  \bibfield  {author} {\bibinfo {author} {\bibfnamefont {W.}~\bibnamefont {Liang}}, \bibinfo {author} {\bibfnamefont {F.}~\bibnamefont {Ticozzi}},\ and\ \bibinfo {author} {\bibfnamefont {G.}~\bibnamefont {Vallone}},\ }\bibfield  {title} {\bibinfo {title} {Optimizing measurements sequences for quantum state verification},\ }\href@noop {} {\bibfield  {journal} {\bibinfo  {journal} {Quantum Information Processing}\ }\textbf {\bibinfo {volume} {22}},\ \bibinfo {pages} {419} (\bibinfo {year} {2023})}\BibitemShut {NoStop}%
\bibitem [{dat()}]{data_av}%
  \BibitemOpen
  \href@noop {} {}\bibinfo {note} {The data used to generate figures 3,4,6 are available together with paper source files at https://doi.org/10.48550/arXiv.2509.24636}\BibitemShut {NoStop}%
\bibitem [{\citenamefont {Jacques}\ \emph {et~al.}(2009)\citenamefont {Jacques}, \citenamefont {Neumann}, \citenamefont {Beck}, \citenamefont {Markham}, \citenamefont {Twitchen}, \citenamefont {Meijer}, \citenamefont {Kaiser}, \citenamefont {Balasubramanian}, \citenamefont {Jelezko},\ and\ \citenamefont {Wrachtrup}}]{Jacques09}%
  \BibitemOpen
  \bibfield  {author} {\bibinfo {author} {\bibfnamefont {V.}~\bibnamefont {Jacques}}, \bibinfo {author} {\bibfnamefont {P.}~\bibnamefont {Neumann}}, \bibinfo {author} {\bibfnamefont {J.}~\bibnamefont {Beck}}, \bibinfo {author} {\bibfnamefont {M.}~\bibnamefont {Markham}}, \bibinfo {author} {\bibfnamefont {D.}~\bibnamefont {Twitchen}}, \bibinfo {author} {\bibfnamefont {J.}~\bibnamefont {Meijer}}, \bibinfo {author} {\bibfnamefont {F.}~\bibnamefont {Kaiser}}, \bibinfo {author} {\bibfnamefont {G.}~\bibnamefont {Balasubramanian}}, \bibinfo {author} {\bibfnamefont {F.}~\bibnamefont {Jelezko}},\ and\ \bibinfo {author} {\bibfnamefont {J.}~\bibnamefont {Wrachtrup}},\ }\bibfield  {title} {\bibinfo {title} {Dynamic polarization of single nuclear spins by optical pumping of nitrogen-vacancy color centers in diamond at room temperature},\ }\href {https://doi.org/10.1103/PhysRevLett.102.057403} {\bibfield  {journal} {\bibinfo  {journal} {Phys. Rev. Lett.}\ }\textbf {\bibinfo {volume} {102}},\ \bibinfo {pages} {057403}
  (\bibinfo {year} {2009})}\BibitemShut {NoStop}%
\bibitem [{\citenamefont {Steiner}\ \emph {et~al.}(2010)\citenamefont {Steiner}, \citenamefont {Neumann}, \citenamefont {Beck}, \citenamefont {Jelezko},\ and\ \citenamefont {Wrachtrup}}]{enhancedNV}%
  \BibitemOpen
  \bibfield  {author} {\bibinfo {author} {\bibfnamefont {M.}~\bibnamefont {Steiner}}, \bibinfo {author} {\bibfnamefont {P.}~\bibnamefont {Neumann}}, \bibinfo {author} {\bibfnamefont {J.}~\bibnamefont {Beck}}, \bibinfo {author} {\bibfnamefont {F.}~\bibnamefont {Jelezko}},\ and\ \bibinfo {author} {\bibfnamefont {J.}~\bibnamefont {Wrachtrup}},\ }\bibfield  {title} {\bibinfo {title} {Universal enhancement of the optical readout fidelity of single electron spins at nitrogen-vacancy centers in diamond},\ }\href {https://doi.org/10.1103/PhysRevB.81.035205} {\bibfield  {journal} {\bibinfo  {journal} {Phys. Rev. B}\ }\textbf {\bibinfo {volume} {81}},\ \bibinfo {pages} {035205} (\bibinfo {year} {2010})}\BibitemShut {NoStop}%
\bibitem [{\citenamefont {Jiang}\ \emph {et~al.}(2009)\citenamefont {Jiang}, \citenamefont {Hodges}, \citenamefont {Maze}, \citenamefont {Maurer}, \citenamefont {Taylor}, \citenamefont {Cory}, \citenamefont {Hemmer}, \citenamefont {Walsworth}, \citenamefont {Yacoby}, \citenamefont {Zibrov},\ and\ \citenamefont {Lukin}}]{Jiang09}%
  \BibitemOpen
  \bibfield  {author} {\bibinfo {author} {\bibfnamefont {L.}~\bibnamefont {Jiang}}, \bibinfo {author} {\bibfnamefont {J.~S.}\ \bibnamefont {Hodges}}, \bibinfo {author} {\bibfnamefont {J.~R.}\ \bibnamefont {Maze}}, \bibinfo {author} {\bibfnamefont {P.}~\bibnamefont {Maurer}}, \bibinfo {author} {\bibfnamefont {J.~M.}\ \bibnamefont {Taylor}}, \bibinfo {author} {\bibfnamefont {D.~G.}\ \bibnamefont {Cory}}, \bibinfo {author} {\bibfnamefont {P.~R.}\ \bibnamefont {Hemmer}}, \bibinfo {author} {\bibfnamefont {R.~L.}\ \bibnamefont {Walsworth}}, \bibinfo {author} {\bibfnamefont {A.}~\bibnamefont {Yacoby}}, \bibinfo {author} {\bibfnamefont {A.~S.}\ \bibnamefont {Zibrov}},\ and\ \bibinfo {author} {\bibfnamefont {M.~D.}\ \bibnamefont {Lukin}},\ }\bibfield  {title} {\bibinfo {title} {Repetitive readout of a single electronic spin via quantum logic with nuclear spin ancillae},\ }\href {https://doi.org/10.1126/science.1176496} {\bibfield  {journal} {\bibinfo  {journal} {Science}\ }\textbf {\bibinfo {volume} {326}},\ \bibinfo
  {pages} {267} (\bibinfo {year} {2009})}\BibitemShut {NoStop}%
\bibitem [{\citenamefont {Neumann}\ \emph {et~al.}(2010)\citenamefont {Neumann}, \citenamefont {Beck}, \citenamefont {Steiner}, \citenamefont {Rempp}, \citenamefont {Fedder}, \citenamefont {Hemmer}, \citenamefont {Wrachtrup},\ and\ \citenamefont {Jelezko}}]{Neumann10b}%
  \BibitemOpen
  \bibfield  {author} {\bibinfo {author} {\bibfnamefont {P.}~\bibnamefont {Neumann}}, \bibinfo {author} {\bibfnamefont {J.}~\bibnamefont {Beck}}, \bibinfo {author} {\bibfnamefont {M.}~\bibnamefont {Steiner}}, \bibinfo {author} {\bibfnamefont {F.}~\bibnamefont {Rempp}}, \bibinfo {author} {\bibfnamefont {H.}~\bibnamefont {Fedder}}, \bibinfo {author} {\bibfnamefont {P.~R.}\ \bibnamefont {Hemmer}}, \bibinfo {author} {\bibfnamefont {J.}~\bibnamefont {Wrachtrup}},\ and\ \bibinfo {author} {\bibfnamefont {F.}~\bibnamefont {Jelezko}},\ }\bibfield  {title} {\bibinfo {title} {Single-shot readout of a single nuclear spin},\ }\href {https://doi.org/10.1126/science.1189075} {\bibfield  {journal} {\bibinfo  {journal} {Science}\ }\textbf {\bibinfo {volume} {329}},\ \bibinfo {pages} {542} (\bibinfo {year} {2010})}\BibitemShut {NoStop}%
\bibitem [{\citenamefont {Marro}\ and\ \citenamefont {Basile}(1994)}]{marroControlledConditionedInvariants1994}%
  \BibitemOpen
  \bibfield  {author} {\bibinfo {author} {\bibfnamefont {G.}~\bibnamefont {Marro}}\ and\ \bibinfo {author} {\bibfnamefont {G.}~\bibnamefont {Basile}},\ }\href@noop {} {\emph {\bibinfo {title} {Controlled and Conditioned Invariants in Linear System Theory}}},\ Vol.~\bibinfo {volume} {30}\ (\bibinfo {year} {1994})\BibitemShut {NoStop}%
\end{thebibliography}%

\end{document}